\def\draft{1}
\documentclass[11pt,letterpaper]{article}
\usepackage{enumerate}
\usepackage{amsthm}
\usepackage{amsmath,latexsym,amssymb,enumerate,graphicx}
\usepackage{hyperref}

\usepackage{fullpage}
\usepackage{color}
\usepackage{cancel}
\usepackage{float}

\usepackage{array}
\newcolumntype{C}{>{$}c<{$}}
\newcolumntype{R}{>{$}r<{$}}
\allowdisplaybreaks

\newcommand{\nc}{\newcommand}

\newcommand{\<}{\langle}
\renewcommand{\>}{\rangle}

\newtheorem{theorem}{Theorem}
\newtheorem{definition}[theorem]{Definition}
\newtheorem{lemma}[theorem]{Lemma}
\newtheorem{fact}[theorem]{Fact}
\newtheorem{corollary}[theorem]{Corollary}

\newcommand{\thmref}[1]{Theorem~\ref{thm:#1}}
\newcommand{\lemref}[1]{Lemma~\ref{lem:#1}}
\newcommand{\myeqref}[1]{\eqref{eq:#1}}
\nc\eq[1]{(\ref{eq:#1})}
\newcommand{\setft}[1]{\mathrm{#1}}
\newcommand{\Density}{\setft{D}}

\newcommand{\NP}{\textsc{NP}}

\newcommand{\MIP}{\textsc{MIP}}
\newcommand{\QMA}{\textsc{QMA}}

\newcommand{\NEXP}{\textsc{NEXP}}
\newcommand{\PSPACE}{\textsc{PSPACE}}

\newcommand{\epr}{{\rm EPR}}

\renewcommand{\a}{\alpha}
\newcommand{\h}{\mathcal{H}}
\newcommand{\bra}[1]{\< #1 |}
\newcommand{\ket}[1]{| #1 \>}

\newcommand{\eps}{\epsilon}
\newcommand{\ot}{\otimes}
\newcommand{\Xlin}{\mathcal{X}} 
\newcommand{\Zlin}{\mathcal{Z}}

\newcommand{\cA}{\mathcal{A}}

\newcommand{\mH}{\mathcal{H}}

\DeclareMathOperator{\E}{\mathbf{E}}
\DeclareMathOperator{\Id}{\mathbb{I}}
\DeclareMathOperator{\poly}{poly}
\DeclareMathOperator{\Tr}{Tr}
\DeclareMathOperator{\CON}{C}
\DeclareMathOperator{\DIS}{d}
\DeclareMathOperator{\Drho}{\DIS_\rho}
\DeclareMathOperator{\Trho}{\Tr_\rho}

\renewcommand{\Re}{\mathrm{Re}}

\ifnum\draft=1
\newcommand{\tnote}[1]{\textcolor{magenta}{\small {\textbf{(Thomas:}
      #1\textbf{) }}}}
\newcommand{\anote}[1]{\textcolor{red}{\small {\textbf{(Anand:} #1\textbf{) }}}}

\else
\newcommand{\tnote}[1]{{}}
\newcommand{\anote}[1]{{}}

\fi
\newcommand{\avg}[2]{\left\langle #1 \right\rangle_{#2}}

\newcommand{\wg}{\omega^*_{\ensuremath{\rm ac}}}
\newcommand{\enc}{\ensuremath{\rm pauli}}
\newcommand{\energy}{\ensuremath{\rm energy}}
\newcommand{\ac}{{\ensuremath{\rm ac}}}
\newcommand{\stab}{\ensuremath{\rm stab}}

\newcommand{\app}[2]{\approx_{#1}^{#2}}

\newcommand{\C}{\mathbb{C}}

\newcommand{\N}{\mathbb{N}}

\newcommand{\sx}{\hat{X}}
\newcommand{\sz}{\hat{Z}}

\renewcommand{\sp}{\hat{H}}
\newcommand{\cx}{\bar{X}}
\newcommand{\cz}{\bar{Z}}

\newcommand{\sigx}{\sigma_{X}}

\newcommand{\sigz}{\sigma_{Z}}

\newcommand{\EPR}{\mathrm{EPR}}
\newcommand{\MS}{\mathrm{MS}}

\newcommand{\Pos}{\setft{Pos}}

\newcommand{\Obs}{\setft{Obs}}

\newcommand{\proj}[1]{\ket{#1}\bra{#1}}
\begin{document}
\title{Robust self-testing of many-qubit states}
\author{Anand Natarajan\thanks{Center for Theoretical Physics,
    MIT, Cambridge, USA. email:\texttt{anandn@mit.edu}. } \qquad Thomas
  Vidick\thanks{Department of Computing and Mathematical Sciences,
    California Institute of Technology, Pasadena, USA. email:
    \texttt{vidick@cms.caltech.edu}.}} 
\date{\today}
\maketitle

\begin{abstract}
  We introduce a simple two-player test which certifies that the players apply tensor products of Pauli $\sigma_X$ and
  $\sigma_Z$ observables on the tensor product of $n$ EPR pairs. The test has constant robustness: any strategy achieving success probability within an additive $\eps$ of the optimal must be $\poly(\eps)$-close, in the appropriate distance measure, to the honest $n$-qubit strategy. The test involves $2n$-bit questions and $2$-bit answers. 
		The  key technical ingredient is a quantum version of the classical
  linearity test of Blum, Luby, and Rubinfeld.

	As applications of our result we give (i) the first robust self-test for $n$ EPR pairs; (ii) a quantum multiprover interactive proof
  system for the local Hamiltonian problem with a constant number
  of provers and classical questions and answers, and a constant
  completeness-soundness gap independent of system size; (iii) a robust protocol for delegated quantum   computation.
\end{abstract}

% \thispagestyle{fancy}
% \renewcommand{\headrulewidth}{0pt}
% \rhead{MIT-CTP/4846}
% \lhead{}

\section{Introduction}

Quantum non-local games lie at the intersection of several areas of
quantum information. They provide a natural approach to
\emph{device-independent certification} or \emph{self-testing} of
unknown quantum states. Device-independent
certification has applications to quantum cryptography, from quantum
key distribution~\cite{VV14,MS14} to delegated
computation~\cite{ReichardtUV13nature,FitzsimonsH15}. The key idea
behind these applications is that certain nonlocal games, such as the
CHSH game~\cite{CHSH69}, provide natural statistical tests that can be
used to certify that an arbitrary quantum device implements a certain
``strategy'' specified by local measurements on an entangled state
(e.g. an EPR pair).

A common weakness of all existing self-testing results is that their
performance scales poorly with the number of qubits of the state that
is being tested. Given a self-test, define (somewhat informally) its
robustness as the largest $\eps=\eps(\delta)$ such that a success
probability at least $\omega^*_{\rm opt}-\eps$ in the test certifies the
target state up to error (in trace distance and up to local isometries) at most $\delta$, where
$\omega^*_{\rm opt}$ is the success probability achieved by an ideal
strategy. All previously known tests for $n$-qubit states required
$\eps \ll \poly(\delta,1/n)$.%, with an often high degree polynomial
%dependence on both $\eps$ and $n$.

Our main result is a form of robust self-test for {any} state that can
be characterized via expectation values of tensor products of standard Pauli
$\sigma_X$ or $\sigma_Z$ observables. (This includes a tensor product of $n$ EPR pairs; see below.)

\begin{theorem}[simplified\footnote{The complete statement of
the theorem says much more, and provides a characterization of near-optimal strategies.}]\label{thm:main_informal}
  Let $\mathcal{P}$ be a set of $n$-qubit observables, each of
  which is a tensor product of single-qubit Pauli $\sigma_X$, $\sigma_Z$ or
  $\pm I$, and $\lambda_{\rm max} = \| \E_{P \sim \mathcal{P}} [\, P\,]\|$.  For any $\eta \geq 0$ there exists a
  $p = p(\eta)=\Theta(\eta^c)$, where $0<c<1$ is a universal constant,
  and a $7$-player nonlocal game with $O(n)$-bit questions and
  $O(1)$-bit answers such that
	$$\omega^*_{\rm opt}\,
 =\, \frac{1}{2} + p\, \lambda_{\rm max}\,\pm\, \eta.$$
\end{theorem}

We view the theorem as a robust self-test in the following
sense. Suppose a many-qubit state $\ket{\psi}$ can be characterized as
the leading eigenvector of an operator $O=\E_{P \sim \mathcal{P}} [P]$
obtained as the average of $n$-qubit Pauli operators, with associated
eigenvalue $\lambda_{\rm max}\in[-1,1]$. For example, if $\mathcal{P}$
is the uniform distribution over
$\{\sigma_X \otimes \sigma_X,\sigma_Z\otimes\sigma_Z\}^{\otimes n}$
then $\lambda_{\rm \max}=1$ and the leading eigenvector is the tensor
product of $n$ EPR pairs. More generally, if $H$ is a local
Hamiltonian with $m$ local $XZ$ terms we can take $\mathcal{P}$ to be
$\Id$ with probability $1/2$ and the negation of a random term of $H$
with probability $1/2$. Then
$\lambda_{\rm max} = \frac{1}{2}-\frac{1}{2m}\lambda_{\rm min}(H)$ and
the leading eigenvector is a ground state of $H$.

Theorem~\ref{thm:main_informal} provides a nonlocal game such that the
optimal success probability in the game is directly related to
$\lambda_{\rm max}$, thereby providing a test distinguishing between
small and large $\lambda_{\rm max}$. In fact the complete statement of
the theorem (see Theorem~\ref{thm:main} in Section~\ref{sec:manyprovers}) says much
more. In particular, we provide a complete characterization (up to
local isometries) of strategies achieving a success probability at
least $\omega^*_{\rm opt} - \eps$, for $\eps$ sufficiently small but independent of $n$, showing that such strategies must be based on a particular
encoding (based on a simple, fixed error-correcting code) of an eigenvector associated to $\lambda_{\rm max}$. 

\subsection{Applications}

Before giving an overview of the proof of the theorem we discuss some consequences of the theorem that help underscore its generality.

\paragraph{Hamiltonian complexity.} 
A first consequence of Theorem~\ref{thm:main_informal} is that the ground state
energy of a local Hamiltonian can be certified via a non-local game
with questions of polynomial length and constant-length
answers.%up to \emph{constant} error.

\begin{corollary}\label{cor:qma}
  Let $H$ be an $n$-qubit Hamiltonian that can be expressed as a weighted sum, with real coefficients, of tensor
  products of $\sigma_X$ and $\sigma_Z$ operators on a subset of the
  qubits, and normalize $H$ such that $\|H\|\leq 1$. Suppose it is given that $\lambda_{\rm min}(H) \leq a$ or
  $\lambda_{\rm min}(H) \geq b$ for some $0\leq a<b\leq 1$. There exists a one-round interactive
  proof protocol between a classical polynomial-time verifier and
  $7$ entangled provers where the verifier's (classical)
  questions are $O(n /(b-a))$ bits long, the provers' (classical)
  answers are $O(1)$ bits each, and the maximum probability that the
  verifier accepts is
$$\lambda_{\rm min} \leq a \implies \omega^*_{\rm opt} \geq p_{\rm
  c} := \frac{1}{2} +2\,\eta_0 ,\qquad \lambda_{\rm min} \geq
b \implies \omega^*_{\rm opt} \leq p_{\rm s} := \frac{1}{2} + \eta_0,$$
where $\eta_0>0$ is a small (universal) constant. 
\end{corollary}

Since the class of Hamiltonians considered in Corollary~\ref{cor:qma}
is QMA-complete~\cite{CM13}, the corollary can be viewed as a
quantum analogue of the (games variant of the) exponentially long PCP
based on the linearity test of Blum, Luby and
Rubinfeld~\cite{BLR93}.
%\footnote{See~\cite[Section 5]{AharonovAV13qpcp} for a very different approach to the same question.} 
	Indeed, observe that the game constructed in the
corollary has an efficient verifier, polynomial-length questions, and
a constant completeness-soundness gap $\eta_0$ as soon as the original promise
on the ground state energy for the Hamiltonian exhibits an
inverse-polynomial completeness-soundness gap. The derivation of
Corollary~\ref{cor:qma} from Theorem~\ref{thm:main_informal} involves a step of
gap amplification via tensoring, and relies on the fact that
Theorem~\ref{thm:main_informal} allows any $XZ$-Hamiltonian with no
requirement on locality.

A similar result to Corollary~\ref{cor:qma} was obtained by
Ji~\cite{ji2015classical}, and we build on Ji's techniques. The results
are incomparable: on the one hand, the question size in our protocol
is much larger ($\poly(n)$ bits instead of $O(\log n)$
for~\cite{ji2015classical}); on the other hand, the dependence of the
verifier's acceptance probability on the ground state energy is much
better, as in~\cite{ji2015classical} the completeness-soundness gap
remains inverse polynomial.

\paragraph{An exponential quantum PCP.} 
The expert reader may already have noted that the complexity-theoretic
formulation of Corollary~\ref{cor:qma} described above already follows
from known results in quantum complexity. Indeed, recall that the
class $\QMA$ is in $\PSPACE$, and that single-round multiprover
interactive proof systems for $\PSPACE$ (and even $\NEXP$) follow from
the results in~\cite{IV12,Vidick13xor}. Another possible proof
approach for the same result could be obtained by repeating the
protocol in~\cite{ji2015classical} a polynomial number of times;
provided there existed an appropriate parallel repetition theorem this
would amplify the soundness to a constant (although the answer length
would now be polynomial). In fact, based on a recent result by Ji~\cite{ji16nexp} it seems likely that both approaches, based either on our results or parallel repetition of~\cite{ji16nexp}, could lead to an exponential ``quantum-games'' PCP for all languages in $\NEXP$ (instead of just $\QMA$). Even though in purely complexity-theoretic terms the result would still not be new, we believe that the techniques from Hamiltonian complexity developed to obtain it show good promise for further extensions. 

Indeed our protocol has some advantages over the generic sequence of known reductions. One is efficiency: in our protocol the provers merely need
access to a ground state of the given local Hamiltonian and the
ability to perform constant-depth quantum circuits. It is this
property that enables our application to delegated quantum computing
(see below for more on this). Answers in our protocol are a constant number of bits; the reductions mentioned above would require soundness amplification via parallel repetition, which would lead to answers of (at least) linear length. 

Even though they may not provide the most immediately compelling application of
  Theorem~\ref{thm:main_informal}, the complexity-theoretic consequences of
  Corollary~\ref{cor:qma} tie our results to one
  of their primary motivations, the \emph{quantum PCP
    conjecture}. Broadly speaking, the quantum PCP research program is
  concerned with finding a robust analog of the Cook-Levin theorem for
  the class $\QMA$. The ``games variant'' of this conjecture states
  that estimating the optimal winning probability of entangled players
  in a multiplayer nonlocal game, up to an additive constant, is
  $\QMA$-hard. In other words, that there exists an $\MIP^*$ protocol
  for $\QMA$ with $O(\log(n))$-bit messages and constant
  completeness-soundness gap. The best progress to date in this
  direction is the work of Ji~\cite{ji2015classical}, which gives a
  five-prover one-round $\MIP^*$ protocol with $O(\log(n))$-bit
  messages for the local Hamiltonian problem such that the verifier's
  maximum acceptance probability is $a-b\lambda_{min}(H)n^{-c}$ for
  positive constants $a,b,c$. This falls short of the games PCP
  conjecture in that the completeness-soundness gap is inverse
  polynomial in $n$, rather than constant.\footnote{Here again we point the interested reader to the recent~\cite{ji16nexp}, which obtains a protocol with similar parameters, involving $8$ provers, for all languages in $\NEXP$.}

%In trying to improve Ji's result, one encounters two main difficulties. First, since
%known $\QMA$-complete versions of the local Hamiltonian problem have
%inverse-polynomial promise gap, it is reasonable to expect that some stage of the protocol should
%perform ``gap amplification'' while maintaining the
%locality of the Hamiltonian. Indeed, this is strategy
%adopted by Dinur's proof of the PCP
%theorem~\cite{Dinur07pcp}. However, Dinur's proof extensively relies
%on copying bits of the classical proof, and straightforward attempts
%to quantize it run into the barrier of the no-cloning
%theorem. However, there is also a second difficulty unique to the
%quantum case: even assuming that the underlying local Hamiltonian
%instance has constant promise gap, the protocol of Ji still has an
%inverse-polynomial completeness-soundness gap. In contrast, in the
%classical setting, there exists simple protocols for CSPs like the clause-variable game
%\anote{CITES} that are ``gap preserving.'' This work can be viewed as
%making progress on both of these difficulties, without completely
%solving them. 
Our results suggest an approach to the problem from a different angle: we provide a ``gap preserving'' protocol, in the sense that the completeness-soundness gap is a polynomial function of the underlying promise gap of
the Hamiltonian, but \emph{independent} of the system size
$n$. However, this occurs at the cost of much longer messages --- polynomial instead of logarithmic. 

\paragraph{Dimension witnesses.}
Consider the operator
$O = (\frac{1}{2}(\sigma^X\otimes \sigma^X+ \sigma^Z\otimes
\sigma^Z))^{\otimes n}$.
This operator has largest eigenvalue $1$ with associated eigenvector
$\ket{\epr}^{\otimes n}$, where $\ket{\epr} = \frac{1}{\sqrt{2}}\ket{00}+\frac{1}{\sqrt{2}}\ket{11}$. In this case the proof of
Theorem~\ref{thm:main_informal} allows us to obtain the following robust
self-test for $\ket{\epr}^{\otimes n}$:

\begin{corollary}\label{cor:epr}
 % There exists a universal constant $\delta>0$ such that the following holds. 
For any integer $n$ there is a two-player game with
  $O(n)$-bit questions and $O(1)$-bit answers such that (i) there is a
  strategy with optimal winning probability $\omega^*$ that uses
  $\ket{\epr}^{\otimes n}$ as entangled state; (ii) for any $\eps>0$, any strategy with
  success probability at least $\omega^*-\eps$ must be based on
  an entangled state which is (up to local isometries) within distance $\delta=\poly(\eps)$  of  $\ket{\epr}^{\otimes n}$.
\end{corollary}

The game whose properties are summarized in Corollary~\ref{cor:epr} is based on the CHSH game. By using the Magic Square game instead, it is possible to devise a test with perfect completeness, $\omega^*=1$, which can be achieved using an honest strategy based on the use of $(n+1)$ EPR pairs.

To the best of our knowledge, all prior self-tests for any family of states had a robustness guarantee going to $0$ inverse polynomially fast with the number of qubits tested (see Section~\ref{sec:related} below for a more thorough comparison with related works). 

\paragraph{Delegated computation.}
It was noticed in~\cite{FitzsimonsH15} that an interactive proof
system for the local Hamiltonian problem can also be used for
delegated quantum computation with so-called \emph{post-hoc}
verification. The key idea is to use the Feynman-Kitaev construction
to produce a Hamiltonian encoding the desired computation; measuring
the ground energy of this Hamiltonian reveals whether the computation
accepts or rejects.
%This connection was used in~\cite{FitzsimonsH15} to give a delegated computation scheme with a limited quantum verifier, with only the capability to send, receive, and measure a constant number of qubits. 
Following the same connection, we are able to give a post-hoc
verifiable delegated computation scheme with a purely classical
verifier and a constant number of provers. The provers only need the
power of BQP. The scheme has a constant completeness-soundness gap
independent of the size of the circuit to be computed, unlike the
scheme of~\cite{FitzsimonsH15} and the classical scheme
of~\cite{ReichardtUV13nature}, which both have inverse-polynomial
gaps. However, unlike the scheme of~\cite{ReichardtUV13nature}, our
protocol is not \emph{blind}: the verifier must reveal the entire
circuit to be computed to all the provers before the verification
process starts. We refer to Section~\ref{sec:delegated} for more
details on this application.

\subsection{Proof overview}

The proof of Theorem~\ref{thm:main_informal} builds on ideas from complexity
theory and quantum information. We draw inspiration from
classical ideas in the closely related areas of probabilistically checakble proofs, locally
testable codes, and property testing. The link between these
areas and quantum self-testing is the idea of verifying a \emph{global}
property of an unknown object using only limited measurements. The two most important components of
the proof are a ``locally verifiable'' encoding of arbitrary $n$-qubit
quantum states~\cite{FV14}, and a quantum analogue of the linearity
test of Blum et al.~\cite{BLR93}. Since the second component is the
more novel we explain it first.

\paragraph{Linearity testing of quantum observables.}

The simplest instantiation of the classical PCP theorem relies on the
Hadamard code to robustly encode an $n$-bit string (e.g. an assignment
to an instance of $3$-SAT). Under this code,
a string $u \in \{0,1\}^n$ is encoded as the $2^n$-bit long truth table of the
function $f_u:\{0,1\}^n \to \{-1, 1\}$ given by $f_u(x) = (-1)^{u
  \cdot x}$, where $\cdot$ is the bitwise inner product. The function
$f_u(x)$ is said to be \emph{linear}, since $f_u(x+y) = f_u(x)f_u(y)$.
The key property of the Hadamard code
which makes it useful in this context is that it is locally
testable. A local test is given by the BLR linearity test: given query
access to a function $f:\{0,1\}^n\to\{-1,1\}$, by checking that
$f(x+y)=f(x)f(y)$ at randomly chosen $x,y$ the test certifies that any
$f$ that is accepted with probability at least $1-\eps$ has the form
$f\approx_\eps f_u$ for some $u\in\{0,1\}^n$, where
$f_u:a\mapsto (-1)^{u\cdot a}$ and $\approx_\eps$ designates equality
on an $(1-O(\eps))$ fraction of inputs.

Here is a ``quantum'' reformulation of this test as a nonlocal game:
instead of querying an oracle for $f$ at three points, play a
three-player nonlocal game where each player is asked for the value
at a point. This test is sound even if the players share an entangled quantum
state~\cite{IV12}, but success in the test does not certify quantum behavior:
the players could win with certainty just by sharing a description of
a classical linear function $f_u$; indeed, the main point of the analysis in~\cite{IV12} is precisely to ensure that provers sharing entanglement have no more freedom than to use it as shared randomness in selecting $u$. 

 In contrast, we seek an extension of the test which certifies a very specific type of quantum behavior that could not be emulated by classical means alone: specifically, that the observable $O_x$ measured by a player upon receiving question $x$ itself
is (up to a change of basis, and in the appropriate ``state-dependent'' norm) close to $\otimes_i \sigma_X^{x_i}$.
%---any
%successful strategy should be equivalent under an isometry to
%measuring all qubits of an $n$-qubit state to obtain an $n$-bit string
%$u$, and then returning the value of $f_u$ at the queried point $x$.
We give a test which achieves this. The test performs a combination of
a linearity test in the $X$-basis and a linearity test in the
$Z$-basis; an ``anticommutation game'' (which can be taken to be a
version of the CHSH or Magic Square games) is used to constrain how the
results of the two linearity tests relate to each other.

\begin{theorem}[Pauli braiding test, informal]\label{thm:braiding}
  There exists a two-player nonlocal game, based on the combination of
  (i) a linearity test in the $X$ basis (questions $x\in\{0,1\}^n$);
  (ii) a linearity test in the $Z$ basis (questions $z\in\{0,1\}^n$;
  (iii) an ``anticommutation game'' (based on e.g. the CHSH or Magic Square games) designed to test for generalized   anti-commutation relations (questions $(x,z)\in\{0,1\}^{2n}$), such
  that any strategy
%$\ket{\psi}_{AB}$, $\{P^A_{x,z}\}$, $\{P^B_{x,z}\}$ 
  that has success probability $\omega^*_{\rm opt}-\eps$ for some
  $\eps>0$ must be based on observables $A(x),A(z),A(x,z)$ and an
  entangled state $\ket{\psi}_{AB}$ such that up to local isometries
  $$A(x) \approx_{\delta} \otimes_i \sigma_X^{x_i}, \qquad A(z) \approx_{\delta} \otimes_i \sigma_Z^{z_i},\qquad\text{and}
\qquad \ket{\psi}_{AB} \approx_{\delta} \ket{\epr}_{AB}^{\otimes n},$$
where $\delta=\poly(\eps)$. 
\end{theorem}

Neither the linearity test nor the
anticommutation test \emph{alone} would be  sufficient to achieve the
conclusion: as noted above, the linearity test can be passed even by
classical provers, and our anticommutation test can be fooled if the
provers share just one EPR pair. Rather, it is the guarantees provided
by these tests together that enable us to create a tensor-product
structure in the provers' Hilbert space. 

To gain intuition on the test one may think of it in the following way. A standard approach to self-testing $n$ EPR pairs is to fix a decomposition of the Hilbert space as 
\begin{equation}\label{eq:c2n}
\mH = \C^{2^n} \approx \C^2 \otimes \cdots \otimes \C^2,
\end{equation}
and perform the CHSH (or Magic Square) test ``in parallel'', on each copy of $\C^2$. To the best of our knowledge such test only leads to robustness bounds with a polynomial dependence in $n$. In contrast the test on which Theorem~\ref{thm:braiding} is based relies on the observation that the decomposition~\eqref{eq:c2n} need not be rigidly fixed a priori; indeed there are many bases in which such decomposition of $\C^{2^n}$ in tensor factors can be performed. In particular,  \emph{any} pair of anti-commuting observables on $\mH$ suffices to specify a copy of $\C^2$, on which a CHSH test can in principle be performed (here we crucially rely on rotation invariance of the $2^n$-dimensional maximally entangled state). Our test leverages this observation by performing a CHSH test for each possible pair of Pauli operators ($\sigma_X(a)$, $\sigma_Z(b)$), where $a,b\in\{0,1\}^n$ are such that $a\cdot b=1$. Each of these tests amounts to identifying a copy of $\C^2$ and performing the CHSH test on it. Contrary to the parallel-repeated CHSH test these copies are not independent, and this is what makes our test much more robust. 

\paragraph{Encoding quantum states.}
The second component of the proof of Theorem~\ref{thm:main_informal} is a
procedure, first introduced in~\cite{FV14,ji2015classical}, for
encoding an $n$-qubit quantum state in a constant number of $n$-qubit
shares such that certain properties of the encoded state (such as
expectation values of local Pauli observables) can be verified through
a classical interaction with provers each holding one of the
shares. This is akin to how the ``games'' variant of the classical PCP
theorem is derived from the ``proof-checking'' variant: while in the
classical setting a proof can be directly shared across multiple
provers, in the quantum setting we use a form of secret-sharing code
that allows for distributing quantum information.
% Given a state of interest (such as the ground state of a local
% Hamiltonian), the expectation value of a local $XZ$ Pauli operator
% on the state can be estimated by a combination of tests involving a
% stabilizer-state generalization of the CHSH state~\cite{Ji}, which
% depends on the precise encoding used.

This procedure is efficient in that the total number of tests that can
be performed (equivalently, the number of questions) is polynomial in
$n$. However, the test in~\cite{FV14,ji2015classical} is not robust,
and is only able to provide meaningful results for values of $\eps$
that scale inverse-polynomially with $n$. By extending the Pauli
braiding test, Theorem~\ref{thm:braiding}, to the stabilizer framework
of~\cite{ji2015classical} we obtain a procedure which is meaningful
for constant $\eps$. The drawback is that the provers may now be asked
to measure all their qubits, and questions have length linear in $n$;
however the total effort required of the classical verifier (and of provers given access to the state) remains
polynomial in the size of the instance.

\subsection{Related work}
\label{sec:related}

We build on a number of previous works in quantum information
and complexity theory. 
Motivation for the problem we consider goes back to a question of Aharonov and Ben-Or (personal communication, 2013), who asked how a quantum generalization of the exponential classical PCP could look like if it was not derived through the ``circuitous route'' obtained as the compilation of known but complex results from the theory of classical and quantum interactive proof systems (as described earlier). In this respect we point to~\cite[Section
  5]{AharonovAV13qpcp} for a very different approach to the same question based on a ``quantum take'' on the arithmetization technique. 
	
  More directly, our work builds on the already-mentioned
  works~\cite{FV14, ji2015classical} initiating the study of
  entangled-prover interactive proof systems for the local Hamiltonian
  problem. The idea of using a distributed encoding of the ground
  state in order to obtain a multiprover interactive proof system for
  the ground state energy is introduced in~\cite{FV14}. In that work
  the protocol required the provers to return qubits; the possibility
  for making the protocol purely classical was uncovered by
  Ji~\cite{ji2015classical}. Our use of stabilizer codes, and the
  stabilizer test which forms part of our protocol, originate in his
  work. In addition we borrow from ideas introduced in the study of 
quantum multiprover interactive proofs with entangled provers~\cite{KM03, CHTW04}, and especially the three-prover linearity test of~\cite{IV12} and the use of oracularization from~\cite{IKM09} to make it into a two-prover test. 
		
% Compared to the works mentioned above, and~\cite{FV14, ji2015classical} in particular, our result differs in two important
%   respects, making it incomparable in general. First, the question
%   size in our protocol is much larger: $\poly(n)$ bits instead of $O(\log
%   n)$ for~\cite{ji2015classical}. Second, the dependence of the
%   verifier's acceptance probability on the ground state energy is much
%   better: while our dependence is of a constant factor,
%   in~\cite{ji2015classical} there is a polynomial
%   scaling.\footnote{One could attempt to recover our result by
%     repeating the protocol in~\cite{ji2015classical} a polynomial
%     number of times. Provided there existed an appropriate parallel
%     repetition theorem, this would amplify the soundness to a
%     constant. However, the answer length would now be polynomial, and
%     it is unclear whether this could be reduced to a constant without
%     having to go once more through the complicated reductions
%     of~\cite{Vidick13xor}, defeating the purpose.}  
% Interpreting all three results as steps towards a quantum PCP
% theorem,~\cite{FV14,ji2015classical} propose a first step that is
% \emph{size-preserving} (the number of questions is polynomial in the
% instance size) but has only an inverse polynomial gap; in contrast we
% take the route of a \emph{gap-preserving} construction, but the number
% of questions becomes exponential in the instance size. 

Our results are related to work in quantum self-testing, in particular testing EPR
pairs~\cite{McKagueYS12rigidity} and more general entangled
states~\cite{mckague2014self}. A sequence of results has established that the presence of $n$ EPR pairs between two provers can be certified via a protocol
using queries and answers of length polynomial in $n$, with
inverse-polynomial completeness-soundness gap. This was first
achieved by~\cite{ReichardtUV13nature} for a test based on serial
repetition of the CHSH game, and subsequently by~\cite{McKague15} for
a single-round test based on CHSH, by~\cite{OV16} for an XOR game
based on CHSH, and by~\cite{CN16}
and~\cite{Coladangelo16} independently for the parallel-repeated Magic
Square game. Viewed in the context of these results our work is the only one to provide a test whose robustness does not depend on the number of EPR pairs being tested. The reason this can be achieved is the linearity test(s) performed as part of the Pauli braiding test, which we see as a major innovation of our work.

% However, his protocol is not
% directly comparable to ours since it requires $\poly(n)$-bit answers
% and $\log(n)$-bit questions, whereas we use $O(1)$-bit answers and
% $\poly(n)$-bit questions; in addition and more importantly for us the
% completeness-soundness gap in his protocol scales polynomially with
% $n$, whereas in our case the scaling is independent of
% $n$.\tnote{yes, this is a big difference, so I think we shouldn't
% write his results are ``similar to ours'' above}

\subsection{Open questions and future directions}
\label{sec:open}

In our opinion the most important direction for future work is to improve the efficiency of the Pauli braiding test in terms of the number of questions required. Can the test be derandomized, to questions of sub-linear, or even logarithmic, length? Such a result would establish the main step left towards proving the games variant of the quantum PCP conjecture. Instead of directly derandomizing the current test, can it be made more robust, perhaps using some of the ideas based on low-degree polynomial encodings that are key to the classical PCP theorem? 

Aside from this challenging problem, there are several open questions that we find
interesting and may be more approachable.
\begin{enumerate}
\item In the classical PCP setting, the Hadamard code and the BLR
  linearity test can be used for \emph{alphabet reduction}: converting
  a PCP or $\MIP$ protocol with large answer alphabet into one with
  a binary alphabet. This is a key step in Dinur's proof of the PCP
  theorem~\cite{Dinur07pcp}. Can the linearity test also be used for
  alphabet reduction of $\MIP^*$ protocols? The difficulty is to preserve completeness; if the optimal honest strategy uses a maximally entangled state then the adaptation should be straightforward, but if not it may be more challenging --- perhaps ideas similar to our protocol for ground states of $XZ$ Hamiltonians can be used. 
%\item The Pauli braiding test does not have perfect completeness. It
%  would be desirable to develop an analogue having that property. This
%  could be achieved by replacing the use of the CHSH test by the magic
%  square game, perhaps building on the analysis
%  in~\cite{Coladangelo16, CN16}. 
\item 
  An obvious application for many EPR pairs is quantum key
  distribution (QKD). A major contribution
  of~\cite{ReichardtUV13nature} was to show that the sequential self-test for
  many EPR pairs obtained in that paper could be leveraged into a
  scheme for quantum key distribution (QKD) that is 
  secure in the device-independent (DI) model of security. We believe it should be possible to use the Pauli braiding test to develop a DIQKD protocol in which the interaction with the devices can be executed in parallel, but we leave this possibility for future work.

%Consider the following outline for a protocol. Eve prepares two
  %arbitrary devices, hands one to Alice and the other to Bob. Alice
  %instructs her device to implement the test implied by
  %Corollary~\ref{cor:epr} on a random subset of half its qubits, and
  %to measure the other half in either the computational or Hadamard
  %bases, chosen at random. She communicates her choice of subset and
  %queries to Bob, who instructs his device to implement complementary
  %queries: as indicated in the test for the first subset, and
  %identical to Alice's for the remainder. Each device receives all its
  %queries at once, and malicious devices may implement an arbitrary
  %measurement depending simultaneously on all queries; there is no a
  %priori assumption that the device is made of $n$ qubits.
%
  %Based on the proof of Corollary~\ref{cor:epr} and the details of the
  %test we hope that it should be possible to prove that such protocol
  %leads to a linear key rate and tolerates a constant fraction of
  %noise in the devices. The main technical challenge is to show that
  %\emph{conditional} on the Pauli test passing on the random subset
  %chosen by Alice, measurements on the remaining half of the ``qubits''
  %are consistent with a state that is very close to $n/2$ EPR
  %pairs. While we do not expect this approach, even if successful, to
  %match the current state of the art in terms of key rates for DIQKD,
  %a possible advantage over existing schemes is that the whole
  %protocol can proceed in a single round.

\item The energy test can be viewed as a ``device independent property
  test'' for any property of a quantum state that can be suitably
  expressed as a Hamiltonian. Are there other 
  device-independent property tests that can be formulated in our framework? It would be interesting to see
  which results from the survey of Montanaro and de
  Wolf on quantum property testing~\cite{montanaro2013survey} can be
  generalized to the device-independent setting.
\end{enumerate}

\paragraph{Organization of the paper.}
In Section~\ref{sec:prelim} we introduce some notation used throughout
as well as basic definitions of stabilizer codes and local
Hamiltonians. In Section~\ref{sec:linearity}, we establish an important technical component of our results, the linearity test and its quantum
analysis. We expand this into a two-prover self-test for the Pauli group on
$n$-qubits in Section~\ref{sec:twoprover}, which forms the basis for our main result. In
Section~\ref{sec:manyprovers} we extend this test to handle more than
two provers and show how it can be combined with
an energy measurement test to devise a game for the local
Hamiltonian problem. In Section~\ref{sec:delegated} we
discuss the application of our protocol to delegated computation. 
%We close with a discussion of open problems in Section~\ref{sec:open}.

%----------------------%
\section{Preliminaries}
\label{sec:prelim}
%----------------------%

We assume basic familiarity with quantum information but give all required definitions. We refer to the standard textbook~\cite{NieChu01} for additional background material. 

\subsection{Quantum states and measurements}

A $n$-qubit pure quantum state is represented by a unit vector
$\ket{\psi}\in \C^2 \otimes \cdots \otimes \C^2 = (\C^2)^{\otimes n}
\approx \C^{2^n}$,
where the ket notation $\ket{\cdot}$ is used to signify a column
vector. A bra $\bra{\psi}$ is used for the conjugate-transpose
$\bra{\psi} = \ket{\psi}^\dagger$, which is a row vector. We use
$\|\ket{\psi}\|^2 = |\bra{\psi}\psi\rangle|$ to denote the Euclidean
norm, where $\bra{\psi}\phi\rangle$ is the skew-Hermitian inner
product between vectors $\ket{\phi}$ and $\ket{\psi}$. A $n$-qubit
mixed state is represented by a density matrix, a positive
semi-definite matrix $\rho \in \C^{2^n} \times \C^{2^n}$ of trace
$1$. The density matrix associated to $\ket{\psi}$ is the rank-1
projection $\ket{\psi}\bra{\psi}$. We use $\Density(\mH)$ to denote the set of all density matrices on $\mH$. 

For a matrix $X$, $\|X\|$ will refer to the operator norm, the largest
singular value. When the Hilbert space can be decomposed as
$\h = \h_A \otimes \h_B$ for some $\h_A$ and $\h_B$, and $X$ is an
operator on $\h_A$, we often write $X$ as well for the operator
$X\otimes\Id_{\h_B}$ on $\h$. It will always be clear from context
which space an operator acts on. All Hilbert spaces considered in the paper are finite dimensional. 

We use $\Pos(\mH)$ to denote the set of positive semidefinite operators on $\mH$. 
A $n$-qubit measurement (also called POVM, for projective
operator-valued measurement) with $k$ outcomes is specified by $k$
positive matrices $M=\{M_1,\ldots,M_k\} \subseteq \Pos(\C^{2^n})$
such that $\sum_i M_i = \Id$. The measurement is \emph{projective} if
each $M_i$ is a projector, i.e. $M_i^2 = M_i$. The probability of
obtaining the $i$-th outcome when measuring state $\rho$ with $M$ is
$\Tr(M_i \rho)$. By Naimark's dilation theorem, any POVM can be
simulated by a projective measurement acting on an enlarged state;
that is, for every POVM $M = \{M_i\}_i$ acting on state
$\ket{\psi} \in \mathcal{H}$ there exists a projective measurement
$M' = \{P_i\}_i$ and a state
$\ket{\psi}\otimes \ket{\phi} \in \mathcal{H} \otimes
\mathcal{H}_{\text{ancilla}}$
with the same outcome probabilities as $M$. Moreover, the
post-measurement state after performing $M$ is the same as the
\emph{reduced} post-measurement state obtained after performing $M'$
and tracing out the ancilla subsystem $\mathcal{H}_{\text{ancilla}}$.

An $n$-qubit observable is a Hermitian matrix
$O\in \C^{2^n}\times\C^{2^n}$ that squares to identity. We use $\Obs(\mH)$ to denote the set of observables acting on $\mH$. $O\in\Obs(\mH)$ is
diagonalizable with eigenvalues $\pm 1$, $O=P_+-P_-$, and
$P=\{P_+,P_-\}$ is a projective measurement. For any state $\rho$,
$\Tr(O\rho)$ is the expectation of the $\pm 1$ outcome obtained when
measuring $\rho$ with $P$. If $\rho = \ket{\psi}\bra{\psi}$ we
abbreviate this quantity,
$\Tr(O\rho) = \Tr(P_+\rho)-\Tr(P_-\rho) = \langle \psi | O | \psi
\rangle$ as $\avg{P}{\psi}$.

A convenient orthogonal basis for the real vector space of $n$-qubit observables is given by the set $\{I,\sigma_X,\sigma_Y,\sigma_Z\}^{\otimes n}$, where $\{I,\sigma_X,\sigma_Y,\sigma_Z\}$ are the four single-qubit Pauli observables
\begin{equation}\label{eq:pauli-def}
 I = \begin{pmatrix} 1 & 0 \\ 0 & 1 \end{pmatrix},\quad
\sigma_X= \begin{pmatrix} 0 & 1 \\ 1 & 0 \end{pmatrix}, \quad \sigma_Y = \begin{pmatrix} 0
  & -i \\ i & 0\end{pmatrix}, \quad \sigma_Z = \begin{pmatrix} 1 & 0 \\ 0 &
  -1\end{pmatrix} .
	\end{equation}
We call the eigenbasis of $\sigma_X$ (resp. $\sigma_Z$) the $X$-basis (resp. $Z$-basis). 
We 
often consider operators that are tensor products of just
$I$ and $\sigma_X$, or just $I$ and $\sigma_Z$. We denote these by $\sigma_X(a), \sigma_Z(b)$, where
the strings $a, b \in \{0, 1\}^n$ indicate which qubits to apply the
$\sigma_X$ or $\sigma_Z$ operators to: a $0$ in position $i$ indicates an $I$ on
qubit $i$, and a $1$ indicates an $\sigma_X$ or $\sigma_Z$. 
We denote by 
$$\ket{\EPR} = \frac{1}{\sqrt{2}} \ket{0}\ket{0} + \frac{1}{\sqrt{2}}\ket{1}\ket{1}$$
the unique state stabilized by both $\sigma_X\otimes\sigma_X$ and
$\sigma_Z\otimes \sigma_Z$.

\subsection{Stabilizer codes}
\label{sec:stabilizer}

Stabilizer codes are the quantum analogue of linear codes. For an introduction to the theory of stabilizer codes we refer to~\cite{Gottesman97}. We will only use very elementary properties of such codes. 

The codes we consider are \emph{Calderbank-Shor-Steane (CSS)
  codes}~\cite{CalderbankShor96,Steane96}. For an $r$-qubit code the
codespace, the vector space of all valid codewords, is the subspace of
$(\C^2)^{\otimes r}$ that is the simultaneous $+1$ eigenspace of a set
$\{S_1,\ldots,S_k\}$ of $r$-qubit pairwise commuting Pauli observables
called the stabilizers of the code. The stabilizers form a group under
multiplication. Unitary operations, such as a Pauli $X$ or $Z$
operators, on the logical qubit are implemented on the codespace by
logical operators $X_{logical}$ and $Z_{logical}$. The smallest CSS code is Steane's $7$-qubit code~\cite{Steane96}. Table~\ref{tab:stabilizer} lists a set of stabilizers that generate the stabilizer group of the code. 

\begin{table}
  \centering
  \begin{tabular}[h]{l | C C C C C C R }
    &1 & 2 & 3 & 4 & 5 & 6 & 7 \\ \hline
    Stabilizers &I & I & I & \sigma_X & \sigma_X  & \sigma_X & \sigma_X \\
    &I & \sigma_X & \sigma_X & I & I & \sigma_X & \sigma_X \\
    &\sigma_X & I & \sigma_X & I & \sigma_X & I & \sigma_X \\
    &I & I & I & \sigma_Z & \sigma_Z  & \sigma_Z & \sigma_Z \\
    &I & \sigma_Z & \sigma_Z & I & I & \sigma_Z & \sigma_Z \\
    &\sigma_Z & I & \sigma_Z & I & \sigma_Z & I & \sigma_Z \\ \hline
    Logical X & \sigma_X & \sigma_X & \sigma_X & \sigma_X & \sigma_X & \sigma_X & \sigma_X \\
    Logical Z & \sigma_Z & \sigma_Z & \sigma_Z & \sigma_Z & \sigma_Z & \sigma_Z & \sigma_Z
  \end{tabular}
  \caption{Stabilizer table for the 7-qubit Steane code}
  \label{tab:stabilizer}
\end{table}
		
Every CSS code satisfies certain properties which will be useful for
us. Firstly, both the stabilizer generators and the logical operators
can be written as tensor products of only $I$, $\sigma_X$, and $\sigma_Z$ operators --- there are no
$\sigma_Y$. This simplifies our protocol, allowing us to consider only two
distinct basis settings. Secondly, every CSS code has the following symmetry: for every index $i \in [r]$ there
exists stabilizers $S_X$, $S_Z$ such that $S_X$ is a tensor product of
only $\sigma_X$ and $I$ operators and has an $\sigma_X$ at position
$i$, and $S_Z$ is equal to $S_X$ with all $\sigma_X$ operators replaced by
$\sigma_Z$ operators. 

These properties imply the following simple observation,
which will be important for us. For every Pauli operator $P\in\{I,\sigma_X,\sigma_Z\}$ acting on the $i$-th qubit of the code there is a tensor product $\bar{P}$ of Paulis acting on the remaining $(r - 1)$ qubits such that $P \otimes \bar{P}$ is a stabilizer
operator on the whole state, and moreover each term in the tensor
product is either identity or $P$. Indeed, the choice of $\bar{P}$ is not unique. Henceforth, we use the
notion $\bar{P}$ to denote \emph{any} such operator, unless
otherwise specified. 

\subsection{Local Hamiltonians}\label{sec:local-h}

A $n$-qubit local Hamiltonian is a Hermitian, positive semidefinite  operator $H$ on $(\C^2)^{\otimes {n}}$ that can be decomposed as a sum $H = \sum_{i=1}^m H_i$ with each $H_i$ is local, i.e. $H_i$ can be written as $H_i = I\otimes\cdots I \otimes h_i \otimes I \otimes \cdots \otimes I$, where $h_i$ is a Hermitian operator on $(\C^2)^{\otimes k}$ with norm (largest singular value) at most $1$. The smallest $k$ for which $H$ admits such a decomposition is called the locality of $H$. The terms are normalized such that $\| H_i\| \leq 1$ for all $i$. A family of Hamiltonians $\{H_i\}$ acting on increasing numbers of qubits is called local if all $H_i$ are $k$-local for some $k$ independent of $n$ (for us $k$ will always be $2$).

The local Hamiltonian problem is the prototypical $\QMA$-complete problem, as 3SAT is for $\NP$. 

\begin{definition}
Let $k\geq 2$ be an integer. 
  The $k$-\emph{local Hamiltonian problem} is to decide, given a
  family of $k$-local Hamiltonians $\{H_n\}_{n\in\N}$ such that $H_n$ acts on $n$ qubits, and functions $a,b:\N\to (0,1)$ such that $b - a = \Omega( \poly^{-1}(n))$, if the smallest eigenvalue
  of $H_n$ is less than $a(n)$ or greater than $b(n)$.
\end{definition}

Here we restrict our attention to Hamiltonians 
\[ H = \frac{1}{m} \sum_{i=1}^{m} H_i, \]
for which each term $H_i$ can be written as a linear combination of
tensor products of Pauli $I$, $\sigma_X$ and $\sigma_Z$ observables only. Such
Hamiltonians are known to be QMA complete for some constant $k$ (see Lemma~22
of~\cite{ji2015classical} for a proof).
% ; in particular we consider a
% restricted class of $2$-local Hamiltonians, called Hamiltonians of
% $XZ$ form, for which each $H_i$ can be
% written as
% $H_i = \alpha_{i_1 i_2}(X_{i_1} \otimes X_{i_2} + Z_{i_1} \otimes
% Z_{i_2})$,
% where $i_1,i_2\in\{1,\ldots,n\}$ indicate the qubits on which a Pauli
% $X$ or $Z$ acts and the coefficients $\alpha_{i_1i_2}\in\R$ satisfy
% $|\alpha_{i_1i_2}|\leq 1$.

% \begin{theorem}[Cubitt and Montanaro~\cite{CM13}, lemma
%   21]
% The local Hamiltonian problem for $XZ$-Hamiltonians is QMA-complete.\footnote{In~\cite[Lemma 21]{CM13} this is stated for the $XY$ Hamiltonian, to which $XZ$ is equivalent by local rotation.}
% \end{theorem}

\subsection{State-dependent distance measure and approximations}

We make extensive use of a state-dependent distance
between measurements that has been frequently used in the context of
entangled-prover interactive proof systems (see
e.g.~\cite{IV12,ji2015classical}). For $\rho$ a positive semidefinite matrix and $X$ any linear operator define
$$ \Trho(X) = \Tr(\rho X).$$ 
For any two operators $S, T$, define the
\emph{state-dependent distance} between $S$ an $T$ on a $\rho$ as
$$ \Drho(S, T) := \sqrt{\Trho\big( (S-T)^\dagger(S-T)\rho\big)}.$$
Based on the state-dependent distance we define a distance between POVMs, given by
summing the state-dependent distance between the square roots of the POVM elements.
Let $\{M^a\}$ and $\{N^a\}$ be two POVMs
with the same number of outcomes, indexed by $a$, and let
$\ket{\psi}$ be a quantum state. Then the state-dependent distance between the POVMs
$M$ and $N$ on $\rho$ is denoted as $\Drho(\sqrt{M}, \sqrt{N})$ and
defined as
\begin{align*}
  \Drho(\sqrt{M}, \sqrt{N}) = \Big(\sum_a \Drho\big(\sqrt{M^a}, \sqrt{N^a}\big)^2\Big)^{1/2}.
\end{align*}
While this notation is ambiguous (since the sum over outcomes is not explicitly
indicated), context will always make it clear which notion of $\Drho$ is
intended. We will also drop the square roots in the case of POVMs that are
projective measurements.

To simplify the notation, let $A^a = \sqrt{M^a}$ and $B^a =
\sqrt{N^a}$. Then this distance can be rewritten as:
\begin{align}
  \Drho(\sqrt{M}, \sqrt{N})^2 &= \sum_a \Trho\big( (A^a - B^a)^2 \rho\big) \notag\\
               &= 2 - \sum_a \Re \Trho\big( A^a B^a \big),\label{eq:dpsi-1a}
\end{align}
where we  used the fact that $A^a$ and $B^a$ are
Hermitian and their squares sum to identity. If we specialize to the case of projective measurements
with binary outcomes, we get the following relations (here $A =
A^1 - A^{-1}$ and $B = B^1 - B^{-1}$
are the observables associated to the measurements):
\begin{align}
\Drho(\sqrt{M}, \sqrt{N})^2  &= 2 - \Trho\big(A^{1} B^{1} + A^{-1} B^{-1} + B^{1} A^{1} + B^{-1}
                 A^{-1}\big)  \nonumber \\
               &= 2 - \frac{1}{4} \Trho\big((\Id + A)(\Id + B)
                 + (\Id - A)(\Id - B) + (\Id + B)(\Id + A)
                 + (\Id  - B)(\Id - A)\big) \ket{\psi} \nonumber \\
               &= 2 - \frac{1}{4}\Trho\big(4\Id + 2 AB + 2 BA\big)\nonumber \\
               &= 1 - \frac{1}{2} \Trho\big(AB + BA\big)  \nonumber \\
               &= \frac{1}{2} \Trho\big((A - B)^2\big)\label{eq:dist_observables}    \\
                &= \Drho(A, B)^2. \nonumber
\end{align}
This distance measure has the following useful property:

\begin{lemma}
  Let $\rho$ be positive semidefinite, $C$ be a linear operator such that $\|
  CC^\dagger\| \leq K$ and $S, T$ linear operators. Then
  $$\Big|\Trho (C S)- \Trho(CT ) \Big| \leq
  \sqrt{2K} \Drho(S, T).$$

  Likewise, if $\{C_a\}$ a family of operators such that $\|\sum_a C_aC_a^\dagger\|\leq K$ and $\{M^a\}$
  and $\{N^a\}$ POVMs. Then 
  $$\Big|\sum_a \Trho\big(  C_a \sqrt{M^a}\big) - \sum_a \Trho\big( C_a \sqrt{N^a}\big)\Big| \leq \sqrt{K}\Drho(\sqrt{M}, \sqrt{N}).$$
  \label{lem:approx}
\end{lemma}

\begin{proof}
The proof of both results is identical, and uses the Cauchy-Schwarz inequality;
we show only the proof of the second. Let $A^a = \sqrt{M^a}$ and $B^a = \sqrt{N^a}$. Applying the Cauchy-Schwarz inequality,
  \begin{align*}
    \Big|\sum_a \Trho\big( C_a (A^a - B^a) \big) \Big| 
    &\leq     \Big|\Trho\Big(\sum_a  C_aC_a^\dagger  \Big)\Big|^{1/2} \Big|\Trho\Big(\sum_a (A^a -      B^a)^2  \Big) \Big|^{1/2} \\
    &\leq \sqrt{K} \Drho(\sqrt{M}, \sqrt{N}),
  \end{align*}
	as claimed.
\end{proof}
A second measure of proximity that is often convenient is the
\emph{consistency}. As before, let $\{M^a\}$ and $\{N^a\}$ be POVMs
with the same number of outcomes. Then their consistency is defined
as
\[ \CON_\rho(M, N) = \Re \Big(\sum_a \Trho\big( M^a N^a \big)\Big), \]
so that by~\eqref{eq:dpsi-1a} we have 
\begin{equation} \label{eq:condist}\Drho(\sqrt{M}, \sqrt{N})^2 = 2-2
\CON_\rho(\sqrt{M},\sqrt{N}).\end{equation}
 For collections of binary observables $\{A(a)\}$ and $\{B(a)\}$ we use
\begin{align*}
 \CON_\rho(A,B) &= \sum_a \Re \Big( \sum_{c\in\{0,1\}} \frac{1}{4} \Trho \big( (\Id+(-1)^cA(a))(\Id+(-1)^c B(a) )\big)\Big)\\
&= \frac{1}{2}\Re \Big(1+\sum_a \Trho\big(A(a)B(a)\big)\Big).
\end{align*}
%We also use the notation
%$$ \langle A,B \rangle_\psi = \bra{\psi} A B \ket{\psi}$$
%for any operators $A,B$. The following lemma relates the consistency and the state-dependent
%distance.
%\begin{lemma}[Lemma~10 in~\cite{ji2015classical}]
  %Let $\ket{\psi}$ be a state and $\{M^a\}$ and $\{N^a\}$ be two POVMs with equal numbers of
  %outcomes. If $\CON_\rho(M, N) = 1 - \delta$ for some $\delta\geq 0$ then $\Drho(\sqrt{M},
  %\sqrt{N}) = O(\sqrt{\delta})$.
  %\label{lem:condist}
%\end{lemma}
A useful property of the consistency is that if $M$ and $N$ are POVMs
acting on two separate subsystems of $\rho$, applying Naimark
dilation to each of them results in projective measurements $M'$ and
$N'$ and a state $\rho'$ such that $\CON_\rho(M, N) =
\CON_{\rho'}(M', N')$. 

Given two observables $A$ and $B$, the product $AB$ is an observable
if and only if $A$ and $B$ commute. The following lemma shows how to define
a ``product'' observable $C$ when $A$ and $B$ commute only
approximately in state-dependent distance, such that the action of $C$
on the state is close to $AB$ (and $BA$).
\begin{lemma}
  Let $\rho$ be a density matrix and $A, B$ observables such that
  $\Drho(AB,BA) \leq \delta$ for some $\delta\geq 0$. Let $C$ be the observable defined by
  \[ C = \frac{AB + BA}{|AB + BA|} , \]
	 where we use the convention
  that $M/ |M|$ is defined as the identity on the kernel of $M$. Then
  $$\max\Big\{ \Drho(C,AB),\,\Drho(C,BA)\Big\}\, \leq\, \frac{\sqrt{2}}{2}\,\delta.$$
  \label{lem:jointobs}
\end{lemma}

\begin{proof}
  It is clear from the definition that $C$ is Hermitian and an
  observable (i.e. all its eigenvalues are $\pm 1$). Evaluate
  \begin{align*}
   \Drho(AB, C)^2 &=  2 - \Trho\Big( \frac{AB + BA}{|AB + BA|} AB + BA \frac{AB
      + BA}{|AB + BA|} \Big).
			\end{align*}
 Notice that $AB$ and $BA$ both commute with $(AB + BA)$ and
  hence with $(AB + BA)/|AB + BA|$. Thus the above
    expression simplifies to
		\begin{align*}
   \Drho(AB,C)^2   &= 2 - \Trho\Big( \frac{(AB + BA)^2}{|AB + BA|} \Big) \\
    &\leq 2 - \Trho|AB + BA| \\
    &\leq 2 - \frac{1}{2} \Trho\big( (AB + BA)^2 \big) \\
    &= 2 - \frac{1}{2} \Trho\big(2 \Id + ABAB + BABA\big). 
  \end{align*}
  From the assumption,  $\Drho(AB,BA)^2 = \Trho(2\Id - ABAB -
  BABA) \leq \delta^2$. Substituting in the above, we get $\Drho(AB,C)^2 \leq \delta^2 / 2$, as desired
\end{proof}

Our calculations will often require estimates of the form $\E_{x} \Drho(
A_x, B_x )^2 = O(\eps)$ where the expectation is taken according to some distribution on $x$ (always over a finite set) that will be clear from context. We introduce the following notation to represent the same estimate: 
$$ A_x\ket{\psi} \app{\eps}{x} B_x\ket{\psi}.$$
Here $\ket{\psi}$ can be understood as any purification of $\rho$, with the usual convention that operators are extended to act as identity on spaces on which they are not defined.  
If the symbol $x$ is omitted then the distribution should be clear from context. If it needs to be specified we may write e.g. $ A_x\ket{\psi} \app{\eps}{x|x_1=0} B_x\ket{\psi}$, meaning that the distribution on $x$ is the one clear from context (typically, uniform on $\{0,1\}^n$), conditioned on the first bit of $x$ being a $0$. Although the notation can be ambiguous when taken out of context we hope that it will help make some of the more cumbersome derivations more transparent. 

\subsection{Nonlocal games}
\label{sec:nonlocal}
In the paper we formulate a number of tests meant to be executed between a verifier and $r$ players (sometimes also called provers), for $r\geq 1$ an integer. These tests all take the form of a classical one-round interaction: the verifier samples an $r$-tuple of questions and sends one question to each player; the players each provide an answer to the verifier, who decides to accept or reject. If the verifier accepts the players are said to win the game. 

We call a tuple $(N,\ket{\psi})$, where $\ket{\psi} \in \mH_1\otimes \cdots\otimes \mH_r$ is an entangled state on the joint space of all $r$ players, and $N$ a collection of POVM for each player and possible question to the player, a \emph{strategy} for the players in $G$. Note that we may always assume $\ket{\psi}$ is a pure state and all POVM are projective. 

Given a game $G$ we denote by $\omega^*(G)$ the highest probability of winning that can be achieved by $r$ players sharing quantum entanglement. For a more thorough introduction to nonlocal games in a similar framework as used here we refer to e.g.~\cite{ji2015classical}. 

One of our tests uses nonlocal games as a means to enforce anticommutation relations between a player's observables. Towards this we introduce the following definition. 

\begin{definition}[Anticommutation game]\label{def:ac-game}
Let $\wg\in (0,1]$ and $\delta:[0,1]\to [0,1]$ a continuous function such that $\delta(0)=0$. 
A two-player game $G$ is called a $(\wg,\delta)$ anticommutation game if
$\omega^*(G)\geq \wg$ and moreover there exists questions $q_X,q_Z$ (called
\emph{special questions}) to the second player and $\{\pm1\}$-valued functions $f_X,f_Z$ defined on the player's set of possible answers to questions $q_X,q_Z$ respectively such that the following two properties hold:
\begin{enumerate}
\item {\em Completeness:} There exists a strategy using the state
$\ket{\EPR}_{AB}^{\ot m}$ for some $m \geq 1$ and projective measurements that achieves the optimal success probability $\wg$, and is uch that measurement operators $\{A_q^a\}
\in \Pos(\mH_A)$ for the second player satisfy $\sum_a f_X(a) A_{a_X}^a = \sigma_X \ot \Id$ and $\sum_a f_Z(a) A_{q_Z}^a
= \sigma_Z \ot \Id$, where $\sigma_X$ and $\sigma_Z$ act on
the first EPR pair and the identity on the remaining EPR pairs. Moreover,
for every quesion $q$ received by the second player and answer $a$, the projector $A_q^a$ can be written as
$A_q^a = \sum_j \Pi_j$ where each $\Pi_j$ is the projector onto an eigenspace of
a tensor product of $\sigma_X, \sigma_Z$ and $\Id$.\footnote{This seemingly ad-hoc condition is needed for the use of the anticommutation game in the Hamiltonian self-test described in Section~\ref{sec:manyprovers}, but not in the Pauli braiding test from Section~\ref{sec:twoprover}.} We call such a strategy an
\emph{honest strategy} for $G$.
\item
Soundness: Let a projective strategy for the players in $G$ be given such that the strategy uses entangled state $\ket{\psi}\in\mH_A \otimes \mH_B$ and measurement operators $\{A_q^a\} \in\Pos(\mH_A)$ for the second player. Then for any $\eps>0$, provided the strategy has success probability at least $\wg-\eps$ in $G$, there exists 
 isometries $U:\mH_A \to \C^2\otimes \mH_{A'}$ and $V:\mH_B \to \C^2 \otimes \mH_{B'}$ and a state $\ket{\psi'}\in\mH_{A'}\otimes \mH_{B'}$ such that if 
\begin{equation}\label{eq:def-ac-xz}
X = \sum_a f_X(a) A_{q_X}^a\qquad \text{and}\qquad Z = \sum_a {f_Z(a)} A_{q_Z}^a
\end{equation}
then 
$$ \big\| U\otimes V \ket{\psi} - \ket{\EPR}\otimes \ket{\psi'} \big\|\leq \delta\; \text{and}\;
\max\Big\{ \DIS_{\rho}(X,U^\dagger (\sigma_X\otimes \Id_{A'}) U),\,
\DIS_{\rho}(Z,V^\dagger (\sigma_Z\otimes \Id_{B'}) V)\Big\}\leq \delta.$$
\end{enumerate}
\end{definition}

The CHSH game~\cite{CHSH69} and the Mermin-Peres Magic Square game~\cite{mermin1990simple,peres1990incompatible} are both known to be anti-commutation games. For the former, see e.g.~\cite{McKagueYS12rigidity} and for the latter,~\cite{wu2016device,CN16}. The advantage of the CHSH game is that there is an optimal strategy which only requires a single EPR pair of entanglement. The Magic Square has the advantage of having value $1$, but an optimal strategy requires two EPR pairs. 

\begin{lemma}
The CHSH game is a $(\cos^2 \pi/8,O(\sqrt{\eps}))$ anticommutation game. 
The Magic Square game is a $(1,O(\sqrt{\eps}))$ anticommutation game. 
\end{lemma}

%A fact that will be relevant later is that for both game there exists an optimal
%strategy for the players which only requires the measurement of $\sigma_X$ and
%$\sigma_Z$ Pauli observables, as well as (for the case of CHSH) $(\sigma_X \pm
%\sigma_Z)/\sqrt{2}$.

%----------------------%
\section{The linearity test}
\label{sec:linearity}
%----------------------%

We state and analyze a variant of the classic 3-query linearity test of Blum, Luby, and Rubinfeld~\cite{BLR93} (BLR) that can be played with two entangled players. The two-player test is based on the idea of oracularization with a dummy question introduced in~\cite{IKM09}.  Our analysis builds on~\cite{IV12}, who analyze a $3$-player variant. Their proof is an extension of the Fourier-analytic proof due to H{\aa}stad to the matrix-valued setting. We analyze the two-player variant using similar techniques. 

We note that the use of two players, rather than three as in
the original test, is essential for our applications to
self-testing. Ultimately we will require the provers to succeed in a linearity test performed in either of two mutually incompatible bases (e.g. the $X$ and $Z$ bases). Two provers
can achieve this by sharing a maximally entangled state, but there is no {tripartite} state that would allow three entangled provers to obtain consistent answers whenever they measure their share of the state in either the $X$ or the $Z$ basis. (Formulated differently, $\sigma_X\otimes \sigma_X$ and $\sigma_Z\otimes \sigma_Z$ share a common $+1$ eigenvector, the EPR pair; $\sigma_X\otimes \sigma_X \otimes \sigma_X$ and $\sigma_Z\otimes \sigma_Z \otimes \sigma_Z$ do not. This is a manifestation of entanglement monogamy.)

We show the result in two
steps. First we show that any set of quantum observables satisfying
linearity relations approximately in expectation can be ``rounded'' to a
nearby set of observables satisfying these relations exactly. 

\begin{theorem}
  Suppose there exist observables $\{A(a)\}_{a\in
  \{0,1\}^n}$ in $\Obs(\mH)$ acting on a state
  $\rho \in \Density(\mH)$ such that 
	\begin{equation}\label{eq:qblr-0}
	\E_{a, b} \Trho\big( A(a) A(b) A(a+b)\big) \geq 1-\delta.
	\end{equation}
	Then there exists an extended state $\rho' = \rho \otimes
  \proj{\rm{anc}} \in \Density(\mH\otimes \mH')$ and observables $\{\cA(a)\}$ in $\Obs(\mH\otimes\mH')$ such that 
	\begin{equation}\label{eq:qblr-1}
	\cA(a) \cA(b) = \cA(a+b) \quad \forall a, b \in \{0, 1\}^n\qquad \text{and}\qquad \E_a
  \DIS_{\rho'} (\cA(a), A(a))^2 \leq\delta.
	\end{equation}
\label{thm:qblr}
\end{theorem}

Here, and throughout this paper, the notation $a + b$ denotes the
bitwise XOR of $a$ and $b$, i.e. the sum of $a$ and $b$ viewed
as elements of the additive group ${\mathbb{Z}}_2^n$. We call observables
$\{\cA(a)\}$ satisfying the first set of relations in~\eqref{eq:qblr-1} \emph{exactly linear}.

\begin{proof}
For every $u\in\{0,1\}^n$ consider the Fourier transform  $\hat{A}^u =
\E_{a} (-1)^{a\cdot u} A(a)$. Define measurement operators
  $B^u = (\hat{A}^u)^2$.
By Parseval's identity, these operators form a POVM. Using Naimark's theorem there exists an ancilla space $\mH'$, $\proj{\rm anc}\in\Density(\mH')$, and a projective
measurement $\{C^u\}$ on $\mH\otimes \mH'$ that simulates $\{B^u\}$. Introduce 
observables
\[ \cA(a) = \sum_{u} (-1)^{u\cdot a} C^u. \]
From the orthogonality of the projectors $C^u$ it follows that $\cA(a)
\cA(b) = \cA(a+b)$. Write 
\begin{align*}
  \E_{a} \CON_{\rho'}(A(a), \cA(a))
  &= \frac{1}{2} + \frac{1}{2}\E_{a}\Re\Tr_{\rho'}\big( A(a)\cA(a)\big)\\
	&= \frac{1}{2} + \frac{1}{2}\E_{a}\Re\Big(\sum_u \Trho\big( (-1)^{u\cdot a} A(a)(\hat{A}^u)^2\big)\Big)\\
	&= \frac{1}{2} + \frac{1}{2}\sum_u \Trho\big((\hat{A}^u)^3\big).
\end{align*}
To conclude, note that
  $\sum_u \Trho((\hat{A}^u)^3) =\E_{ab} \Trho(A(a) A(b) A(a + b))$, and use the assumption made in the theorem and the relation between $\CON_{\rho'}$ and $\DIS_{\rho'}^2 $.
\end{proof}

Next we exhibit a two-player game such that any strategy which succeeds with
probability at least $1 - \eps$ in the game must satisfy the assumption~\eqref{eq:qblr-0} of
Theorem~\ref{thm:qblr} for some $\delta = O(\sqrt{\eps})$.

\begin{figure}[H]
\rule[1ex]{16.5cm}{0.5pt}\\
The verifier performs the following one-round interaction with two
players. He starts by choosing one of the players at random and
labels her Alice; the other player is labeled Bob. In each test each player is sent a pair of $n$-bit strings. The $n$-bit strings are always assumed to be sent in lexicographic order.  
\begin{enumerate}
\item Choose two strings $a,b \in \{0,1\}^n$ uniformly at random. Send $(a,b)$ to Alice.
\item Let $c$ be with equal
  probability either $a$, $b$, or $a+b$, and let $c'\in\{0,1\}^n$ be chosen uniformly at random. Send $(c,c')$ to Bob.
	\item The players reply with $\alpha,\beta\in\{\pm 1\}$ and $\gamma,\gamma'\in\{\pm 1\}$ respectively. Depending on the value of $c$ the verifier performs one of the following two tests:
	\begin{enumerate}
        \item
          \emph{Consistency test}: if $c = a$ (resp. $b$), accept if and only if both players return the same value as their answer to that question: $\gamma=\alpha$ (resp. $\gamma=\beta$).
        \item
          \emph{Linearity test}: if $c = a + b$, accept if and only if $\gamma = \alpha \beta$.
        \end{enumerate}
\end{enumerate}
\rule[1ex]{16.5cm}{0.5pt}\\
\caption{The two-player linearity test}
\label{fig:linearity}
\end{figure}

\begin{theorem}
  Suppose two players Alice and Bob succeed in the linearity test of Figure~\ref{fig:linearity} with  probability at least $1 - \eps$, using a shared state $\ket{\psi}_{AB}\in\mH_A\otimes \mH_B$ and projective measurements $\{M_{a,b}^{\a, \beta}\}_{\a,\beta}\in\Pos(\mH_A)$ and $\{N_{a,
    b}^{\a, \beta}\}_{\a,\beta}\in\Pos(\mH_B)$ respectively. Consider the POVM $\{\tilde{M}_a^\alpha\}_\alpha$
  whose elements are given by $\tilde{M}_{a}^{\a} := \E_{b}
  \sum_{\beta} M_{a,b}^{\a,\beta}$, and let $\{A_a^\a\}_{\a} \in \Pos(\mH_A \otimes \mH_{A'})$ be the
  projective measurement obtained by Naimark dilation of $\tilde{M}$.
	
	Then the
  observables $A(a) := A_a^{0} - A_a^{1}$ satisfy 
	$$\E_{a,b} \Tr_{\rho'}(A(a) A(b) A(a+b)) =
  1-O(\sqrt{\eps}),$$
	where $\rho' = \proj{\psi}\otimes \proj{\rm anc}_{\mH_{A'}}$.
  \label{thm:qblr_game}
\end{theorem}

\begin{proof}[Proof of Theorem~\ref{thm:qblr_game}]
Introduce the following conditional measurement operator on $\mH_B$,
\begin{align*}
  N_{a|a b}^{\a} &= \sum_{\beta} N_{ab}^{\a \beta}.
\end{align*}
Note that for every $a,b$ and $\a$, $N_{a|ab}^\a$ is a projector since we assumed each $M_{ab}^{\a \beta}$ is as well. 
Suppose that the players' acceptance probability conditioned on the verifier performing the consistency part of the test (i.e. $c=a$ or $c=b$) is $1 -
\eps_{c}$, while conditioned on the verifier performing the linearity part of the test (i.e. $c=a+b$) it is $1 -
\eps_{l}$, so that $\eps = 2\eps_c/3 + \eps_l/3$. 
Let $\rho = \proj{\psi}_{AB}$. By definition of the consistency test,
\begin{align}
  1 - \eps_{c}  &=  \E_{ab} \CON_\rho\big(\tilde{M}_a, N_{a|ab}\big). \label{eq:lin-test-1a}
	\end{align}
Using Naimark's dilation theorem there is an ancilla space $\mH_{A'}$ and $\proj{\rm anc} \in \Density(\mH_A')$ such that the POVM $\{\tilde{M}^\a_a\}$ acting on $\mH_A$ can be simulated by a projective measurement $\{A_a^\a\}$ acting on ${\rho'} = {\rho} \otimes \proj{\rm anc}_{\mH_{A'}}$. 
Let
$d(a | ab ) = \DIS_{\rho'}\left( A_{a},N_{a|ab } \right)$, so that
by Jensen's inequality, \myeqref{condist} and~\eqref{eq:lin-test-1a},
\begin{align}
  \E_{ab} d(a|ab) &\leq \sqrt{\E_{ab} d(a | ab)^2}\notag\\
	&= O\Big(\sqrt{\E_{ab} \CON_{\rho}(\tilde{M}_a, N_{a|ab})}\Big)\notag \\
                        &= O(\sqrt{\eps_{c}}).\label{eq:lin-test-2}
\end{align}	
Now compute
	\begin{align*}
            \E_{ab} \Tr_{\rho'}\big(A(a) A(b) A(a + b)\big)   &= \E_{ab} \sum_{\a \beta} \Tr_{\rho'}\big(A_{a}^{\a} A_{b}^{\beta} A_{a+b}^{\a \beta} - A_{a}^{\a} A_{b}^{\beta} A_{a + b}^{- \a \beta}\big) \\
            &= 2\E_{ab} \sum_{\a \beta} \Tr_{\rho'}\big(A_{a}^{\a} A_{b}^{\beta} A^{\a \beta}_{a + b}\big) - 1 \\
            &\geq 2\E_{a b} \Big( \sum_{\a \beta}
              \Tr_{\rho'}\big( A^{\a \beta}_{a + b} N_{a|ab}^{\a} N_{b|ab}^{\beta}\big) - O\big( d(a|ab) +
              d(b|ab)\big)\Big)  - 1\\
            &= 1 - O\big(\eps_{l}+ \sqrt{\eps_{c}}\big),
\end{align*}
where the inequality uses \lemref{approx} and the last line is by~\eqref{eq:lin-test-2} and, by definition of the linearity test, 	
	\begin{align}
  1 - \eps_{l} &= \E_{ab} \sum_{\a, \beta} \Trho\big( \tilde{M}_{a + b}^{(\a \beta)} N_{ab}^{\a \beta} \big)\notag \\
  &=  \E_{ab} \sum_{\alpha, \beta} \Trho\big( \tilde{M}_{a+b}^{(\a \beta)} N_{a  |ab}^{\a} N_{b |ab}^{\beta} \big),\notag
\end{align}
since the POVM elements $N_{a|b}^{\a\beta}$ are projectors. 
\end{proof}

%----------------------%
\section{The Pauli braiding test}
\label{sec:twoprover}

In this section we combine the linearity test with an anticommutation test based on any anticommutation game $G_\ac$ satisfying Definition~\ref{def:ac-game} to devise a
two-player test for which the honest strategy consists of applying
tensor products of single-qubit observables in the set $\{\sigma_X(a)\sigma_Z(b),\,a,b\in\{0,1\}\}$. We show that for any strategy with
near-optimal success probability there exists a (local) isometry under which
the players' observables are close (in the state-dependent distance) on average to operators satisfying the Pauli commutation and anti-commutation (``braiding'') relations perfectly. 

\subsection{The protocol}

\begin{figure}[H]
\rule[1ex]{16.5cm}{0.5pt}\\
Let $G_\ac$ be a two-player anticommutation game, with special questions $q_X,q_Z$. 
The verifier performs the following one-round interaction with two
players. He starts by choosing one of the players at random and
labels them Alice; the other player is labeled Bob. In each test a player will be sent a label and a pair of $n$-bit strings. The $n$-bit strings are always assumed to be sent in lexicographic order.  
\begin{enumerate}
\item {\bf Linearity test:} The verifier chooses a basis setting $W \in
  \{X,Z\}$ and sends it to both players. He executes the two-player linearity test with the players.
\item {\bf Anticommutation test:} The verifier chooses two strings $a,b\in\{0,1\}^n$ such that $a\cdot b=1\mod2$ uniformly at random, and sends $(a,b)$ to both players. He executes the game $G_\ac$ with the players and accepts if and only if they succeed. 
\item {\bf Consistency test:} The verifier  chooses two strings
  $a,b\in\{0,1\}^n$ such that $a\cdot b=1\mod2$ uniformly at random, and a basis
  setting $W\in\{X,Z\}$. He sends $(W,a,b)$ to Alice. With probability $1/2$ each, 
\begin{itemize}
\item He samples a question $q$ from the second player's distribution in
  $G_\ac$ and sends $(q,a,b)$ to Bob. If $q=q_X$ (resp. $q=q_Z$) he accepts if and
  only if Alice's answer associated to $a$ (resp. $b$) equals $f_W(\alpha)$, where
  $\alpha$ is Bob's answer and $f_W$ the function from
  Definition~\ref{def:ac-game}. Otherwise, he accepts
  automatically. 
\item He selects a uniformly random $c\in\{0,1\}^n$ and sends $(N,a,c)$ to Bob. He accepts if and only if the product of Alice and Bob's answers associated to the query string $a$ is $+1$. 
\end{itemize}
\end{enumerate}
\rule[1ex]{16.5cm}{0.5pt}
\caption{The two-player Pauli braiding test}
\label{fig:protocol2player}
\end{figure}

The protocol for the Pauli braiding test is described in Figure~\ref{fig:protocol2player}. In the protocol there are several possible types of queries that each player may receive. For convenience we give them the following names: 
\begin{enumerate}
\item A \emph{$W$-query}, represented by $(W, a, b)$, where $W\in\{X,Z\}$ and $a, b $ are uniformly random strings in $\{0,
  1\}^n$. The expected answer is two bits $\alpha,\beta\in\{-1,1\}$. 
%\item A \emph{$Z$-query}, represented by $(Z, a, b)$. Same as an $X$-query, except
%  with $Z$ instead of $X$.
\item A \emph{$G$-query}, represented by $(q, a, b)$ where
  $q$ is a question in $G_\ac$ and $a,b$ are uniformly random strings in $\{0,1\}^n$. The expected answer is a single value $\alpha$ taken from the answer alphabet in $G$. 
\end{enumerate}

To each query is associated an intended behavior of the player, which is specified as part of the \emph{honest strategy} given in the following definition. 

\begin{definition}\label{def:honest-strategy}
The \emph{honest strategy} for the two players in the Pauli braiding test consists of the following. Let $U,V$ be unitaries to an optimal strategy in $G_\ac$ as in Definition~\ref{def:ac-game}, and recall that by the completeness property this strategy can be implemented by sharing $m$ EPR pairs of entanglement. 

The players share the state $\ket{\psi}_{AB} = \ket{\EPR}^{\otimes n}_{AB} \otimes \ket{\EPR}^{\otimes (m-1)}_{A'B'}$.
Upon receiving a query, a player performs the following depending on the type of the query:
\begin{itemize}
\item $W$-query $(W,a,b)$, for $W\in\{X,Z\}$: measure the
  compatible observables $\sigma_W(a)$ and $\sigma_W(b)$ on its share of $\ket{\EPR}^{\otimes n}_{AB}$, and return the two outcomes. 
%\item $Z$-query $(Z,a,b)$: same as $X$-query but with observables $\sigma_Z(a)$ and $\sigma_Z(b)$. 
	\item $G$-query $(q,a,b)$. Suppose the query is sent to Alice, the case of Bob
    being treated symmetrically. Let $W_{a,b}:\C^{2^n}\to\C^{2^n}$ be a unitary such that $W_{a,b}\sigma_X(a)W_{a,b}^\dagger = \Id_{\C^{2^{n-1}}} \otimes \sigma_X$ and $W_{a,b}\sigma_Z(b)W_{a,b}^\dagger = \Id_{\C^{2^{n-1}}} \otimes \sigma_Z$. (Such a $W_{a,b}$ exists and can be agreed upon by the players since in a $G$-query it is always the case that $a\cdot b=1\mod 2$, and both players are sent the same pair $(a,b)$.) Let $\{A_q^\alpha\}_\alpha$  be the projective measurement on $\C^2 \otimes \mH_{A'}$ associated with the first player in a honest strategy in $G$. Then Alice performs the projective measurement 
	$$\big\{(W_{a,b}^\dagger\otimes \Id_{A'})(\Id_{\C^{2^{n-1}}} \otimes A_q^\alpha )(W_{a,b}\otimes \Id_{{A'}})\big\}_\alpha$$
	and returns the outcome. 
\end{itemize}
\end{definition}

Having defined the honest strategy for the players we introduce some notation associated with arbitrary strategies in the protocol. We  specify a strategy using the shorthand $(N,\ket{\psi}_{AB})$. Here $\ket{\psi}_{AB}$ denotes the bipartite state shared by the players, and $N$ the collection of POVM that the players apply in response to the different types of queries they can be asked. Using Naimark's theorem we may assume without loss of generality that $\ket{\psi}_{AB}$ is a pure state and each player's POVM is projective. 

Given a query $(X,a,b)$ (resp. $(Z,a,b)$), we denote by $\{N_{ab}^{\alpha\beta}\}_{\alpha,\beta}$ (resp. $\{M_{ab}^{\alpha\beta}\}_{\alpha,\beta}$) the two-outcome projective measurement that is applied by a given player.
Since the protocol treats the players symmetrically we may assume that these operators are the same for both Alice and Bob (see e.g.~\cite[Lemma~2.5]{Vidick13xor}).
By taking appropriate marginals over the answers we define associated observables for the players, $X^A(a)$ and $Z^A(b)$ for the first player and $X^B(a)$ and $Z^B(b)$ for the second, as 
\begin{equation}\label{eq:hatx-def}
X^A(a) = \frac{1}{2^n}\sum_{b\in\{0,1\}^n} \sum_{\beta\in \{\pm1\}} \big(N_{ab}^{1\beta}-N_{ab}^{-1\beta}\big), \qquad Z^A(b) = \frac{1}{2^n}\sum_{a\in\{0,1\}^n} \sum_{\alpha\in \{\pm1\}} \big(M_{ab}^{\alpha1}-M_{a,b}^{\alpha-1}\big).
\end{equation}
Observables $X^B(a)$ and $Z^B(b)$ for the second player are defined similarly. 

Finally we use ${X'}^A(a,b)$ and ${Z'}^A(a,b)$ to denote the observables defined via~\eqref{eq:def-ac-xz} from Alice's strategy upon questions $(q_X,a,b)$ and $(q_Z,a,b)$ respectively.

\subsection{Statement of results}
\label{sec:twoplayer}

We state the analysis of the Pauli braiding test in two parts: first
we show that success in the test implies that observables~\eqref{eq:hatx-def} constructed from Alice and Bob's measurement operators approximately obey certain relations; then we show that these relations imply the existence of a local isometry under which the operators are close to operators satisfying the relations exactly.  

\begin{theorem}  \label{thm:gametwoplayer}
 Suppose a strategy $(N,\ket{\psi}_{AB})$ succeeds in the
  Pauli braiding test (Figure~\ref{fig:protocol2player}) with probability at least
  $\omega^*_{\enc} - \eps$, when the game $G_\ac$ is an $(\wg,\delta)$ anticommutation game. Then the following approximate relations hold, where operators $W^D$ are defined in~\eqref{eq:hatx-def} for $W\in\{X,Z\}$ and $D\in\{A,B\}$ and $\rho=\proj{\psi}$. 
	\begin{enumerate}
\item (Approximate consistency) For $W \in \{X, Z\}$, $\E_{a}
  \Drho(W^A(a), W^B(a))^2 = O(\eps)$;
\item (Approximate linearity) For $W \in \{X, Z\}$, $\E_{a,b} \Drho
  (W^A(a) W^A(b), W^A(a+b) )^2 = O(\sqrt{\eps})$;
\item (Approximate anticommutation) $\E_{a, b|a\cdot b=1} \Drho( X^A(a) Z^A(b) ,
  - Z^A(b) X^A(a))^{2} = O(\delta(\eps))$;
\item (Approximate commutation) $\E_{a, b|a\cdot b=0} \Drho( X^A(a) Z^A(b) ,
  Z^A(b) X^A(a))^{2} = O(\eps^{1/4} + \delta(\eps)^{1/2})$. 
\end{enumerate}
\end{theorem}

We note that the constant $\omega^*_{\enc}$ is given by
\begin{equation} \omega^*_{\enc} = \frac{2}{3} + \frac{1}{3}
  \wg, \label{eq:omegaenc} \end{equation}
where $\wg \in (0, 1]$ is the winning parameter associated with the $(\wg,\delta)$ anticommutation game $G_\ac$ used in the protocol. Thus if $\wg=1$ then $\omega^*_{\enc}=1$ as well.

\begin{theorem}
\label{thm:isotwoplayer}
Suppose given a bipartite state $\ket{\psi}_{AB}\in\mH_A\otimes\mH_B$, and
observables $\{X^A(a)\}_{a\in\{0,1\}^n}$, $\{Z^A(b)\}_{b\in\{0,1\}^n}$ on $\mH_A$ and
$\{X^B(a)\}_{a\in\{0,1\}^n}, \{Z^B(b)\}_{b\in\{0,1\}^n}$ on $\mH_B$ such that conditions 1.,2. and 3. in Theorem~\ref{thm:gametwoplayer} are satisfied, for some $\eps>0$ and $\delta(\eps)=O(\sqrt{\eps})$.\footnote{The restriction on $\delta$ is not necessary, but it is satisfied for both the CHSH and Magic Square games, and simplifies the presentation.} 
Then there exists a state 
$$\ket{\Psi}_{AB} = \ket{\psi}_{AB} \ot \ket{\EPR}_{A'A''} \ot \ket{\EPR}_{B'B''}\in \big(\mH_A \otimes (\C^2_{A'}\otimes \C^2_{A''})^{\otimes n} \big) \otimes \big(\mH_B \otimes (\C^2_{B'}\otimes \C^2_{B''})^{\otimes n} \big)$$
 and
 observables $\{P^A(a,b)\}$ on $AA'A''$ such that, if $\rho=\proj{\Psi}$, 
\begin{enumerate}
  \item[(a)] (Approximate consistency) $\E_a \Drho(P^A(a,0), X^A(a) \otimes
    \Id_{A'A''})^2 = O(\eps^{1/8})$ and $\E_b
    \Drho(P^A(0,b), Z^A(b) \otimes \Id_{A'A'' })^2 = O(\eps^{1/8})$.
  \item[(b)] (Pauli braiding) For all $a,b,a',b'\in\{0,1\}^n$, $P^A(a,b) P^A(a', b')= (-1)^{a' \cdot b}
    P^A(a+a', b+b') $.
		%= (-1)^{a' \cdot b + b' \cdot a} P^A(a', b')P^A(a, b)$.
\end{enumerate}
Likewise, there exist observables $\{P^B(a,b)\}$ on $BB'B''$ satisfying analogous relations.
\end{theorem}

We note that the Pauli braiding relations expressed in (b) imply the existence
of an isomorphism such that the operators $P^A(a,b)$ (resp. $P^B(a,b)$) are
mapped to ``true'' Pauli operators $\sigma_X^A(a) \sigma_Z^A(b)$
(resp. $\sigma_X^B(a) \sigma_Z^B(b)$). 
% (cf. e.g.~\cite[]{RV}\anote{Undefined citation}\tnote{I was going to cite myself, but it's a bit ridiculous, this is ``well-known''...do you know of a good ref? Otherwise we just claim it}). 

The proofs of Theorem~\ref{thm:gametwoplayer} and Theorem~\ref{thm:isotwoplayer} are given in Sections~\ref{sec:gametwoplayer} and Section~\ref{sec:istwoplayer} respectively. Before moving to the proofs we state an immediate, but powerful, application of the theorems to the problem of establishing dimension witnesses. For this it is sufficient to note the following well-known fact:

\begin{fact}
Let $\rho$ be a density matrix on $\C^{\otimes n}\otimes \C^{\otimes n}$ and $\eps > 0 $ such that 
$$ \frac{1}{2^n} \sum_{P\in\{X,Z\}^n} \Tr( (\sigma_P \otimes \sigma_P) \rho ) \geq 1-\eps,$$
where $\sigma_P = \sigma_{P_1} \otimes \cdots \otimes \sigma_{P_n}$. Then 
$$ \bra{\EPR}^{\otimes n} \,\rho\, \ket{\EPR}^{\otimes n} \geq 1-\eps.$$
\end{fact}

\begin{proof}
Observe that $\ket{\EPR}\bra{\EPR} \geq \frac{1}{2}(\sigma_X\otimes \sigma_X + \sigma_Z \otimes \sigma_Z)$. 
\end{proof}

Combining this fact and Theorems~\ref{thm:gametwoplayer}
and~\ref{thm:isotwoplayer} gives the following consequence: a robust self-test for $n$ EPR
pairs.

\begin{corollary}\label{cor:epr-test}
  Suppose given a strategy $(N,\ket{\psi}_{AB}$) for the players in the Pauli braiding test (Figure~\ref{fig:protocol2player}) with success probability $\omega^*_{\enc} - \eps$, for some $\eps>0$. Then
  there exists a local isometry $\Phi = (\Phi^A:\mH_A\to\mH_{A'}\otimes \mH_{A''},\Phi^B:\mH_B\to\mH_{B'}\otimes \mH_{B''})$ such that 
  $$\Tr\Big( \big(\bra{\EPR}_{A'B'}^{\otimes
      n}\otimes \Id_{A''B''}\big)\,\big( \Phi^A \otimes \Phi^B(\ket{\psi}\bra{\psi}_{AB})\big)\big( \ket{\EPR}_{A'B'}^{\otimes n} \otimes \Id_{A''B''}\big) \Big) = 1 - O\big(\eps^{1/8}\big).$$
\end{corollary}

By instantiating the anticommutation game $G_\ac$ used in the test with the Magic Square game we obtain a robust self-test for $n$ EPR pairs in which the optimal strategy only requires the use of $(n+1)$ EPR pairs and is accepted with probability $1$.\footnote{In fact, for the case of the Magic Square game it is not hard to see that there always exists an optimal strategy in the test using $\max(2,n)$ EPR pairs.}

\subsection{Proof of Theorem~\ref{thm:gametwoplayer}}
\label{sec:gametwoplayer}

The proof of Theorem~\ref{thm:gametwoplayer} proceeds by analyzing each of the three subtests performed in the Pauli braiding test separately, and then putting them together to establish the three conditions claimed in the theorem. We give the proof of the theorem now, assuming the results on each subtest established in \lemref{consistency}, \lemref{anti-commute} and \lemref{commutation} below. 

\begin{proof}[Proof of~\thmref{gametwoplayer}]
 Given a strategy $(N,\ket{\psi}_{AB})$ for the players, define observables $X^A(a),Z^A(b)$ and $X^B(a),Z^B(b)$ as in~\eqref{eq:hatx-def}. %We show that $X^A(a), Z^A(a), X^B(a), Z^B(a)$ satisfy the conclusions of the theorem.
  Property 1. of approximate consistency is established by the consistency test
  (\lemref{consistency}). Property 2. of approximate linearity follows from the 
  Linearity Test (\thmref{qblr_game}). When $a\cdot b = 1 \mod 2$, the
  approximate anticommutation property is
  established by the anticommutation test
  (\lemref{anti-commute}). When $a \cdot b = 0 \mod 2$ the
  corresponding commutation is proved in \lemref{commutation}.
\end{proof}

\subsubsection{Consistency Test}

The following lemma states consequences of the consistency test we will use. 

\begin{lemma}\label{lem:consistency}
Suppose the strategy $(N,\ket{\psi})$ succeeds in the consistency test with probability $1-\eps$. Then there exists $\eps_{\stab} = O(\eps)$ such that 
\[ \E_a \DIS_{\rho}(X^A(a), {X^B}(a) )^2  \leq \eps_{stab} \qquad \text{and} \qquad  \E_b \DIS_\rho({Z^A}(b), {Z^B}(b))^2 \leq \eps_{stab},\]
and
\[ \E_{a,b|a\cdot b=1} \DIS_{\rho}({{X}^A}(a), {{X'}^B}(a,b) )^2  \leq \eps_{stab} \qquad \text{and} \qquad  \E_{a,b|a\cdot b=1} \DIS_\rho({{Z}^A}(b), {{Z'}^B}(a,b))^2 \leq \eps_{stab}. \]
Moreover, the honest strategy succeeds in the test with probability
$1$. 
\end{lemma}

\begin{proof}
It follows from the definition of $\CON_\rho$ that any strategy $(N,\ket{\psi})$ succeeding in the test with probability $1 - \eps$ satisfies 
 \begin{align*}
 \frac{1}{2}\Big( \E_{a,b|a\cdot b=1} \CON_\rho({X}^A(a), {{X'}^B}(a,b)) + \E_{a} \CON_\rho({X}^A(a), {{X}^B}(a))\Big)\,  &= 1 - O(\eps) \\
\frac{1}{2}\Big(   \E_{a,b|a\cdot b=1} \CON_\rho({Z}^A(b), {{Z'}^B}(a,b))+ \E_{b} \CON_\rho({{Z}^A}(b), {{Z}^B}(b))\Big) &= 1 - O(\eps).
 \end{align*}
The first part of the lemma follows directly by applying \myeqref{condist} to the above relations. The second part follows from the definition of the honest strategy and the fact that 
$$\sigma_X \otimes \sigma_X \ket{\EPR} = \sigma_Z \otimes \sigma_Z \ket{\EPR} = \ket{\EPR}.$$ 
%$$\E_{a} \DIS_\rho(X^A(a), {X^B}(a))^2 = O(\eps)\qquad \text{and}\qquad \E_{b} \DIS_\rho({Z^A}(b), {Z^B}(b))^2 = O(\eps),$$
%$$\E_{a} \DIS_\rho(X^A(a), {{X'}^B}(a))^2 = O(\eps)\qquad \text{and}\qquad \E_{b} \DIS_\rho({Z^A}(b), {{Z'}^B}(b))^2 = O(\eps).$$
\end{proof}

\subsubsection{Anticommutation test}

The (approximate) Pauli braiding relations state that 
$$X^A(a){Z^A}(b)\ket{\psi}
\approx (-1)^{a \cdot b} {Z^A}(a)X^A(b) \ket{\psi}.$$
There are two cases: if $a\cdot b = 0 \mod 2$ then the two operators should commute;
otherwise, they should anti-commute. The anticommutation test enforces the latter property. In Section~\ref{sec:commutation} we show how the former can be derived as a consequence.

\begin{lemma}\label{lem:anti-commute}
Suppose the game $G_\ac$ used in the anticommutation test is an $(\wg,\delta)$ anticommutation game. Suppose the strategy $(N,\ket{\psi})$ succeeds in the
anticommutation test with probability $\wg-\eps_\ac$ and in
the consistency test with probability $1 - \eps_{\stab}$. Then  
$$\E_{a, b : a\cdot b = 1} \Drho( X^A(a){Z^A}(b),
(-1)^{a \cdot b} {Z^A}(b)X^A(a) )^2 = O(\delta(\eps_\ac)) + O(\sqrt{\eps_{\stab}}).$$
Moreover, the honest strategy succeeds in this test with
probability $\wg$.
\end{lemma}

\begin{proof}
By definition of the soundness condition of an $(\wg,\delta)$ anticommutation game, the observables ${X'}^A(a,b)$ and ${Z'}^A(a,b)$ satisfy 
$$\E_{a, b : a\cdot b = 1} \DIS_\rho( {X'}^A(a,b){{Z'}^A}(a,b),
(-1)^{a \cdot b} {{Z'}^A}(a,b){X'}^A(a,b) )^2 = O(\delta(\eps_\ac)).$$
Using the triangle inequality, Lemma~\ref{lem:consistency} (note that under the uniform distribution $a\cdot b=1$ with probability at least $1/4$) and Lemma~\ref{lem:approx}, 
$$\E_{a, b : a\cdot b = 1} \DIS_\rho( {X}^B(a){{Z}^B}(b),
(-1)^{a \cdot b} {{Z}^B}(b){X}^B(a) )^2 = O(\delta(\eps_\ac)) + O(\sqrt{\eps_{\stab}}),$$
and analogue relations hold for observables on Alice, using again Lemma~\ref{lem:consistency}. 
\end{proof}

\subsubsection{Commutation}\label{sec:commutation}
The protocol does not involve a test for commutation, as the required property can be derived as a consequence of the existing tests. 

\begin{lemma}\label{lem:commutation}
  Suppose the strategy $(N,\ket{\psi})$ succeeds in the
  linearity and consistency tests with probability at least 
  $1-\eps_{\stab}$ and in the anticommutation test with probability at least $\wg - \eps_\ac$. Then
$$\E_{a,b:a\cdot b=0} \Drho(X^A(a){Z^A}(b) -  {Z^A}(b)X^A(a))^2 =O(\delta(\eps_\ac)^{1/2})+ O({\eps_{\stab}}^{1/4}).$$ 
\end{lemma}

\begin{proof}
We combine the anticommutation, linearity, and consistency tests through the following sequence of approximate identities. Note the approximations are taken under the uniform distribution on $n$-bit strings $a,b$ such that $a\cdot b=0\mod 2$. Since this event occurs with probability at least $1/2$ for uniform $a,b$, the conditioning does not affect any of the approximations used by more than a multiplicative factor $2$. 

Start by applying approximate linearity
(guaranteed by \thmref{qblr_game}) of $Z$ to express $Z(b)$ as a product $Z(c)Z(c+b)$, for uniformly random $c$ such that $c\cdot a =1\mod 2$:
\begin{align*}
X^A(a) Z^A(b)\ket{\psi}   &\approx_{\eps_{\stab}^{1/4}}^{a,b,c| a \cdot b = 0, c\cdot a= 1} X^A(a)
                                             Z^A(c) Z^A(c + b)
                                             \ket{\psi} \\
\intertext{Next use approximate consistency (\lemref{consistency}), to exchange $Z^B(c+b)$ for $Z^A(c+b)$:}
                                        &\app{\sqrt{\eps_{\stab}}}{a,b,c| a \cdot b = 0, c\cdot a= 1} {Z^B}(c+b) X^A(a) Z^A(c) \ket{\psi} \\
  \intertext{Next, apply approximate anticommutation (\lemref{anti-commute}) to anti-commute $X^A(a)$ and $Z^A(c)$:}
                                        &\app{\delta^{1/2}+\eps_{\stab}^{1/4}}{a,b,c| a \cdot b = 0, c\cdot a= 1} -{Z^B}(c +
                                          b) Z^A(c) X^A(a) \ket{\psi}  \\
\intertext{Applying \lemref{consistency} again, transfer $Z^B(c + b)$ back to Alice:}
                          &\app{\sqrt{\eps_{\stab}}}{a,b,c | a \cdot b = 0, c\cdot a = 1} -Z^A(c) X^A(a) Z^A(c+b) \ket{\psi} \\
\intertext{Applying \lemref{anti-commute} anti-commutes $Z^A(c+b)$ and $X^A(a)$:}
                                        &\app{\delta^{1/2}+\eps_{\stab}^{1/4}}{a,b,c| a \cdot b = 0, c\cdot a= 1} Z^A(c) Z^A(c+b)
                                          X^A(a) \ket{\psi} \\
\intertext{Use \lemref{consistency} to transfer $X^A(a)$ to Bob:}
                                        &\app{\sqrt{\eps_{\stab}}}{a,b,c| a \cdot b = 0, c\cdot a= 1} X^B(a) Z^A(c) Z^A(c+b) \ket{\psi} \\
\intertext{Finally apply \thmref{qblr_game} to combine the $Z$ operators, and then \lemref{consistency} to move the $X$ operator back to Alice:}
                          &\app{{\eps_{\stab}}^{1/4}}{a,b,c|a \cdot b = 0, c \cdot a = 1} X^B(a) Z^A(b) \ket{\psi} \\
                          &\app{\sqrt{\eps_{\stab}}}{a,b,c|a \cdot b  = 0, c \cdot a = 1} Z^A(b) X^B(a) \ket{\psi}.
\end{align*}

\end{proof}

\subsection{Proof of Theorem~\ref{thm:isotwoplayer}}
\label{sec:istwoplayer}

We give the proof of~\thmref{isotwoplayer}.

\begin{proof}[Proof of~\thmref{isotwoplayer}]
Adjoin two $n$-qubit registers $A'$ and $A''$ to Alice's system, and
initialize them in the state $\ket{\mathrm{EPR}}_{A'A''}^{\otimes n}$. Define
new observables $X'(a) := X^A(a) \otimes \sigx(a)_{A'} \ot \Id_{A''}$
and $Z'(b) := Z^A(b) \otimes\sigz(b)_{A'} \ot \Id_{A''}$. Further
define observables 
\[ C(a,b) := \frac{X'(a) Z'(b) + Z'(b) X'(a)}{ | X'(a) Z'(b) + Z'(b)
  X'(a) |}, \]
where the notation $| \cdot |$ denotes the matrix absolute
value and we use the convention $0/0=1$. We use the assumptions made in the theorem
(i.e. properties 1, 2 and 3 in Theorem~\ref{thm:gametwoplayer}) to show that $C(a,b)$ satisfies approximate linearity over
$\mathbb{Z}_2^{2n}$, i.e. that $C(a,b) C(a', b')\ket{\Psi} \approx^{a,b,a',b'}
C(a+a', b+b')\ket{\Psi}$.
First, by property 3 (approximate anticommutation),
$X^A(a) Z^A(b) \ket{\Psi} \app{\eps^{1/4}}{a,b} (-1)^{a \cdot b} Z^A(a)
X^A(a)\ket{\Psi}$, and thus $X'(a) Z'(b) \ket{\Psi}\app{\eps^{1/4}}{a,b} Z'(b)X'(a)\ket{\Psi}$. Hence, by \lemref{jointobs} it follows that $C(a,b)\ket{\Psi}
\app{\eps^{1/4}}{a,b} X'(a) Z'(b)\ket{\Psi} $.
Using this relation, we consider the product of two $C$ operators. 
\begin{align*}
   C(a, b) C(a', b') \ket{\Psi} &\app{\eps^{1/8}}{a,b} C(a,b) X^A(a') Z^A(b') \otimes \sigx(a')
                                  \sigz(b') \ket{\Psi}. \\
  \intertext{By property 1 (approximate consistency), we can switch the $X^A$
  and $Z^A$ operators to Bob, and switch the $\sigx, \sigz$ operators to the
  other half of the ancilla. Then, we relate $C(a,b)$ to $X^A(a)Z^A(b)$.}
  &\app{\eps^{1/4}}{a,b,a',b'} Z^B(b') X^B(a') C(a,b) \otimes \sigz(b')_{A''}
    \sigx(a')_{A''} \ket{\Psi} \\
  &\app{\eps^{1/8}}{a,b,a',b'} Z^B(b')X^B(a') X^A(a)Z^A(b) \otimes \sigz(b')_{A''}
    \sigx(a')_{A''} \sigx(a)\sigz(b) \ket{\Psi}. \\
  \intertext{Switching $Z^BX^B$ back to Alice, and $\sigz \sigx$ back to the
  other half of the ancilla,}
  &\app{\eps^{1/4}}{a,b,a',b'} X^A(a) Z^A(b) X^A(a')Z^A(b') \otimes \sigx(a)
    \sigz(b) \sigx(a')\sigz(b') \ket{\Psi}. \\
  \intertext{By the properties of the exact Pauli operators,}
  &=^{a,b,a',b'} (-1)^{a'\cdot b} X^A(a) Z^A(b) X^A(a') Z^A(b') \otimes \sigx(a + a')\sigz(b + b')
    \ket{\Psi}.\\
  \intertext{Applying property 3 (approximate anticommutation),}
  &\app{\eps^{1/8}}{a,b,a',b'} (-1)^{a' \cdot ( b + b')} X^A(a) Z^A(b) Z^A(b') X^A(a') \otimes \sigx(a
    + a') \sigz(b + b') \ket{\Psi}. \\
  \intertext{Applying property 1 (approximate consistency) to
  $X^Z(a')$, and then property 2 (approximate linearity) to combine
  $Z^A(b)$ with $Z^A(b')$, we get}
  &\app{\eps^{1/4}}{a,b,a',b'} (-1)^{a' \cdot (b + b')} X^B(a') X^A(a) Z^A(b + b') \otimes
    \sigx(a + a') \sigz(b + b') \ket{\Psi}. \\
  \intertext{Applying property 3 (approximate anticomutation) to
  $Z^A(b+b')$  and $X^A(a)$,}
  &\app{\eps^{1/8}}{a,b,a',b'} (-1)^{(a+a') \cdot (b+b')} X^B(a') Z^A(b+b')
    X^A(a) \otimes \sigx(a+a') \sigz(b + b') \ket{\Psi}.
    \intertext{Applying property 1 (approximate consistency) to move
    $X^B(a')$ back to Alice, and then applying property 2 (approximate
    linearity) to combine $X^A(a')$ with $X^A(a)$,}
  &\app{\eps^{1/4}}{a,b,a',b'} (-1)^{(a + a') \cdot (b + b')} Z^A(b + b') X^A(a + a') \otimes
    \sigx(a + a') \sigz(b + b') \ket{\Psi}. \\
  \intertext{Finally, applying property 3 (approximate anticommutation) to
  interchange $X^A(a+a')$ and $Z^A(b+b')$,}
  &\app{\eps^{1/8}}{a,b,a',b'} X^A(a + a') Z^A(b + b') \otimes \sigx(a + a')
    \sigz(b + b') \ket{\Psi} \\
  &\app{\eps^{1/8}}{a,b,a',b'} C(a+a', b + b') \ket{\Psi}.
\end{align*}
Applying \thmref{qblr} (over $\{0,1\}^{2n}$), we conclude that there
exist observables $D(a, b)$ acting on an extension of Alice's system
by an ancilla state, satisfying $D(a, b) D(a', b') = D(a+a',
b+b')$ and $\E_{a, b} \Drho(D(a, b), C(a, b))^2 = O(\eps^{1/8})$. Set
\[ P^A(a, b) := D(a, b) \ot \sigx(a)_{A''} \sigz(b)_{A''}. \]
We claim that $P^A(a,b)$ satisfies the desired properties.
\begin{enumerate}
\item[(b)] Pauli braiding: This follows from linearity of $D(a,b)$:
  \begin{align*}
    P^A(a,b) P^A(a',b') &= D(a,b)D(a',b') \ot \sigx(a)_{A''}
                            \sigz(b)_{A''}
                            \sigx(a')_{A''}\sigz(b')_{A''} \\
                          &= D(a+a', b+b') \ot (-1)^{a' \cdot b}
                            \sigx(a + a')_{A''} \sigz(b + b')_{A''} \\
                          &= (-1)^{a' \cdot b} P^A(a+a', b+b') .
  \end{align*}
\item[(a)] Approximate consistency: We establish this in two steps. First,
  note that $D(a, b)$ is approximately consistent with $C(a, b)$, so 
\begin{align*}
   P^A_{a,b} \ket{\Psi} &=^{a,b} D(a,b) \ot \sigx(a)_{A''}
  \sigz(b)_{A''} \ket{\Psi}\\
  &\app{\eps^{1/8}}{a,b} C(a, b) \ot \sigx(a)_{A''}
  \sigz(b)_{A''} \ket{\Psi}\\
  &\app{\eps^{1/4}}{a,b} X^A(a)Z^A(b) \ot \sigx(a)_{A'} \sigz(b)_{A'} \ot \sigx(a)_{A''}
    \sigz(b)_{A''}\ket{\Psi} \\
  &=^{a,b} X^A(a)Z^A(b) \ot \Id_{A' A''} \ket{\Psi},
\end{align*}
where the last line follows since both $\sigma_X\otimes \sigma_X$ and $\sigma_Z\otimes \sigma_Z$ stabilize $\ket{\EPR}$.

Finally, to establish consistency for the operators $P^A(a,0)$ where one
coordinate is fixed to $0$, we exploit the exact Pauli braiding relation:
\begin{align*}
  P^A(a,0) \ket{\Psi} &=^{a,c,d} (-1)^{d \cdot c} P^A(a+c, d) P^A(c,
                        d) \ket{\Psi} \\
  \intertext{By approximate consistency of $P^A$,}
  &\app{\eps^{1/8}}{a,c,d} (-1)^{d \cdot c}  P^A(a + c, d) X^A(c)
    Z^A(d) \ket{\Psi} \\
  \intertext{Applying property 1 (approximate consistency) twice,
  first to $Z^A(d)$ and then to $X^A(c)$, we shift them to Bob's space:}
  &\app{\sqrt{\eps}}{a,c,d} (-1)^{d \cdot c} Z^B(d) X^B(c) P^A(a+c, d)
    \ket{\Psi} \\
  \intertext{Now we apply approximate consistency of $P^A$ again:}
  &\app{\eps^{1/8}}{a,c,d}  Z^B(d) X^B(c) X^A(a+c) Z^A(d) \ket{\Psi} \\
  \intertext{Applying property 3 (approximate anticommutation) to
  $X^A(a+c)$ and $Z^A(d)$, and then property 1 (approximate
  consistency) to $X^B(c)$, we get}
  &\app{\eps^{1/8}}{a,c,d} (-1)^{a \cdot d} Z^B(d) Z^A(d) X^A(a+c)
    X^A(c) \ket{\Psi} \\
  \intertext{We use property 2 (approximate linearity) to combine
  $X^A(a+c)$ and $X^A(c)$:}
  &\app{\eps^{1/4}}{a,c,d}  (-1)^{a \cdot d} Z^B(d) Z^A(d) X^A(a) \ket{\Psi}
  \\
  \intertext{Now, applying property 3 (approximate anticommutation),
  we get,}
  &\app{\eps^{1/8}}{a,c,d} Z^B(d) X^A(a) Z^A(d)  \ket{\Psi}.
  \intertext{Finally, use property 1 (approximate consistency), and
    the fact that $Z^A(d)$ is an observable to get}
  &\app{\sqrt{\eps}}{a,c,d} X^A(a) Z^A(d) Z^A(d) \ket{\Psi} \\
  &=^{a,c,d} X^A(a) \ket{\Psi}.
\end{align*}
\end{enumerate}

\end{proof}

%----------------------%
\section{The Hamiltonian Self-Test}
\label{sec:manyprovers}
%----------------------%

In this section, we build on the Pauli braiding test to construct a test that distinguishes between the cases when a Hamiltonian given as input has ground state energy below, or higher than, pre-specified thresholds (i.e. in the former case the players will have a strategy with high success probability in the protocol, whereas in the latter any strategy will have low success probability).
Due to the nature of our tests we restrict attention to $n$-qubit Hamiltonians specified by a linear combination of $m$ terms, each of which is a tensor product of single-qubit $I, \sigma_X$ or $\sigma_Z$ Pauli operators. 

Recall the Pauli braiding test analyzed in the previous section. As we saw (Corollary~\ref{cor:epr-test}) this test can be used as a robust self-test for an $n$-qubit maximally entangled state. In order to test non-maximally entangled states, we
proceed as in~\cite{FV14,ji2015classical} by requiring the (honest) players to share a qubit-by-qubit encoding of the ground state of the Hamiltonian, where each qubit is encoded using a simple $r$-qubit CSS code. As elucidated
in~\cite{ji2015classical}, any code state, thought of as a bipartite entangled state
across any one of its qubits and the others, is maximally entangled. This allows us to lift the two-player tests which constitute the Pauli braiding test to $r$-player tests, where each player holds one qubit (``share'') of the encoding of each qubit of the ground state, and one of
the players (to be called the \emph{special player}) plays the role of Alice while the remaining $(r - 1)$ players (to be called the \emph{composite player})
play the role of Bob. 

The essential property of the consituent tests of the Pauli braiding test that permit this lifting is that all of the measurements performed by Bob in the honest
strategy can be implemented by measuring the tensor
product of Pauli operators $\sigma_X$, $\sigma_Z$, and $\Id$ on a state of $n$
EPR pairs. (For the anticommutation test, this is ensured by the completeness
condition in Definition~\ref{def:ac-game}, and for the other tests, it can be
seen to hold for both Alice and Bob's measurements). These operators can be implemented
transversally in any CSS code, and moreover in a way such that marginal
distribution of queries received by the special player and each of the composite
players is identical, as we show below.

\subsection{The protocol} 

We describe the protocol in detail. The input is an $n$-qubit
local Hamiltonian $H$ that can be expressed as
\begin{equation}\label{eq:h-form}
H = \frac{1}{m}\sum_{\ell=1}^m H_\ell, \qquad H_\ell = \alpha_\ell \,\sigma_X(a_\ell)\sigma_Z(b_\ell),
\end{equation}
for $\alpha_\ell \in [-1,1]$ and  $a_\ell,b_\ell\in\{0,1\}^n$ such that $a_\ell\wedge b_\ell=0^n$ for all $\ell\in\{1,\ldots,m\}$. 
 The verifier interacts with $r$
players, where $r$ is the number of qubits of codewords in the CSS code chosen for the protocol (such as Steane's $7$-qubit code, as described in Section~\ref{sec:stabilizer},  in which case $r=7$).

\begin{figure}[H]
\rule[1ex]{16.5cm}{0.5pt}\\
Given a local Hamiltonian $H = \sum_{\ell=1}^m \alpha_\ell H_\ell$, where $\alpha_\ell \in [-1,1]$ and each $H_\ell = \sigma_X(a_\ell)\sigma_{Z}(b_\ell)$. Let $p\in(0,1)$ be a parameter of the protocol.

The verifier performs one of the following three tests at random, the first with probability $(1-p)$ and the second and third with probability $p/2$ each. 
\begin{enumerate}
\item (Pauli braiding test) Choose one of the $r$ players uniformly at random to be the
  {special player}. The other players form the {composite
     player}. Simulate the Pauli braiding test with these two
   players, where the role of Alice is assigned to the special player and the role of Bob to the composite player.  
\item (Energy test) Choose $\ell\in\{1,\ldots,m\}$ uniformly at random. Define an operator $Q_\ell$ acting on $rn$ qubits by replacing each $\sigma_X$ in $H_\ell$ with $X_{logical}$ on the $r$-qubit code
  state, and $\sigma_Z$ by $Z_{logical}$. Send each player a query $(XZ,a,b)$ representing the associated share of $Q_\ell$. The players should each return two values in $\{-1,1\}$.\\
	The verifier takes the product of all values received. If its sign disagrees with that of $\alpha_\ell$, he accepts. If they agree, he rejects with probability $|\alpha_\ell|$ and accepts otherwise. 
\item (Energy consistency test) Choose one of the $r$ players uniformly at random to be the
  {special player}. The other players form the {composite
     player}. Let $W\in\{X,Z\}$, each chosen with probability
   $1/2$. Also choose $a, b$ according to the same distribution as in the
   energy test. The verifier performs one of the following tests, each
   chosen with the indicated probability.
\begin{itemize}
\item With probability $1/2$, send the special player $(XZ,a,b)$, and the composite player $\overline{(W, c, c + a)}$ if $W=X$ and $\overline{(W, c, c + b)}$ if $W=Z$, where $c \in \{0, 1\}^n$ is chosen uniformly at random. Accept if the special player's answer agrees with the product of the composite player's two answers.
\item With probability $1/4$, send the special player $(W, c, d)$, and the composite player $\overline{(W, c, c + a)}$, where $c, d \in \{0, 1\}^n$ are chosen uniformly at random. Accept if the special player and composite player agree on the answer associated with $c$.
\item With probability $1/4$, send the special player $(W, c+a,d)$, and the composite player $\overline{(W, c, c + a)}$, where $c, d \in \{0, 1\}^n$ are chosen uniformly at random. Accept if the special player and composite player agree on the answer associated with $c+a$.
\end{itemize}
\end{enumerate}
\rule[1ex]{16.5cm}{0.5pt}\\
\caption{The Hamiltonian self-test}
\label{fig:protocol}
\end{figure}

Although the protocol is to be performed with $r$ ``physical''
players, part of the protocol consists in applying the Pauli braiding test, which is formulated as a two-player test in the previous section. To translate between the $r$ players and the two players in the Pauli braiding test we introduce two ``logical'' players. A query to the logical players (as specified in the Pauli braiding test) is mapped to a query to the $r$ physical players as follows. One of the physical players is chosen at random to play the role of the first logical player (Alice), called the
\emph{special player}. The remaining $(r-1)$ physical players together play the role of the second logical player (Bob), called the \emph{composite player}.\footnote{The physical players remain isolated throughout the protocol and are never allowed to communicate; it is only for purposes of analysis that we group $(r-1)$ physical players into a single logical player. In particular the physical players are never told which logical player they are associated with, and the distribution of queries to any physical player is the same whether it plays the role of the special or composite player.} For a given query $Q$ to the special player of a type among those specified in the Pauli braiding test we define a \emph{complementary query} $\overline{Q}$ for the composite player as per the following lemma.     

\begin{lemma}
  For any $X$-query or $Z$-query, there exists a complementary query $\overline{Q}$ such that
  \begin{enumerate}
  \item The query associated to each physical player forming the composite player in $\overline{Q}$ is of the same type as $Q$. In particular the distribution on query strings is as specified by the query type. 
  \item If all players apply the honest strategy and provide answers $\alpha,\beta$ to $Q$ and  $\overline{\alpha},\overline{\beta}$ to $\overline{Q}$ respectively, where $\overline{\alpha}$ and $\overline{\beta}$ are each obtained as the product of the answer to the corresponding query coming from each of the physical players making up the composite player, it holds that $\alpha\overline{\alpha} = \beta\overline{\beta}=+1$.
  \end{enumerate}
  \label{lem:qbar}
\end{lemma}

\begin{proof}
  Both items follow from the properties of CSS codes described in Section~\ref{sec:stabilizer}. We give the proof for an
  $X$-query $(X, a, b)$. Let the index
  of the special player be $i\in\{1,\ldots,r\}$, and let $S_X$ be a stabilizer
  of the code, such that $S_X$ consists only of $I$ and $\sigma_X$ Paulis and has a $\sigma_X$ in
  position $i$. For each physical
  player $j \neq i$ associated with the composite player, if the operator in
  position $j$ of $S_X$ is $\sigma_X$,  player $j$ is sent the query $(X,
  a, b)$. Otherwise, player $j$ is sent a uniformly random $X$-query
  $(X, c, d)$.
	
		Composite answers $\overline{\alpha},\overline{\beta}$ to the complementary
    query are determined by taking the product of the answers from all
    players who did not receive random strings; using that $S_X$ is a stabilizer of the code  ensures that item 2 is satisfied. 
	
	In the composite query, for a given choice of $S_X$ each player
  receives a query that is either identical to the original query, or
  is a uniformly random string; since the original query is chosen at random this is also the case for each of the physical players associated with the composite player. This proves item 1.
\end{proof}

We can then define associated observables for the players, $\hat{X}(a)$ and $\hat{Z}(b)$ for the special player and $\overline{X}(a)$ and $\overline{Z}(b)$ for the composite player, exactly as in~\eqref{eq:hatx-def}. 

\begin{definition}\label{def:ol-obs}
Let $\{\hat{M}_{a,b}^{\alpha\beta}\}$ (resp. $\{\overline{M}_{a,b}^{\alpha\beta}\}$) be the POVM implemented by the special player (resp. composite player) when asked a query $(W,a,b)$ (resp $(\overline{W},a,b)$; see Lemma~\ref{lem:qbar}), for $W=X$ or $Z$. Define observables
\begin{equation*}
\hat{W}(a) = \frac{1}{2^n}\sum_{b\in\{0,1\}^n} \sum_{\beta\in \{\pm1\}} \big(\hat{M}_{a,b}^{1\beta}-\hat{M}_{a,b}^{-1\beta}\big), \qquad \overline{W}(a) = \frac{1}{2^n}\sum_{b\in\{0,1\}^n} \sum_{\beta\in \{\pm1\}} \big(\overline{M}_{a,b}^{1\beta}-\overline{M}_{a,b}^{-1\beta}\big).
\end{equation*}
\end{definition}

Aside from the Pauli braiding test, the protocol considers two other
tests called the \emph{energy test} and the \emph{energy consistency
  test}. In the energy test, the verifier asks the players to measure
a randomly chosen term in the Hamiltonian. The consistency test is
needed to relate the operators applied in the energy test
to those applied in the Pauli braiding test.
The energy test uses an additional query type, which differs from the types of queries used in the Pauli braiding test:
\begin{enumerate}
\item[3.] An \emph{$XZ$-query} is represented by $(XZ, a, b)$ where $a,b\in\{0,1\}^n$ are  such that $a\wedge b = 0^n$. Note that here, in contrast to $X-$ or $Z-$queries, the strings $a$ and $b$ are ordered. The distribution on $a$ and $b$ depends on the Hamiltonian. The expected answer is two bits $\alpha,\beta\in\{-1,1\}$.
\end{enumerate}

The honest strategy for the players in the Hamiltonian self-test (Figure~\ref{fig:protocol}) consists of applying the honest strategy defined for the Pauli braiding test (Definition~\ref{def:honest-strategy}) whenever the query is of $X$, $Z$, or $G$ type, and the following strategy when it is of $XZ$ type:

\begin{definition}\label{def:honest-h}
In the honest strategy, a player answers an $XZ$-query $(XZ,a,b)$ by measuring the compatible observables $\sigma_X(a)$ and $\sigma_Z(b)$ and returning both outcomes.
\end{definition}

\subsection{Statement of results}

Our main result regarding the Hamiltonian self-test is given in the following theorem, which states the completeness and soundness guarantees of the protocol described in Figure~\ref{fig:protocol}.

\begin{theorem}\label{thm:main}
There exists a constant $0<d<1$ such that the following holds. 
  Let $H$ be a (not necessarily local) Hamiltonian with $m$ terms over $n$ qubits of the form~\eqref{eq:h-form}, and
  $\lambda_{\min}(H)$ the smallest eigenvalue of $H$. Then for every $\eta>0$ there is a choice $p=\Theta(\eta^{1-d})$ for the probability of performing the energy test  in Protocol~\ref{fig:protocol} such that the maximum probability $\omega^*(H)$ with which any $r$-player strategy succeeds in the protocol satisfies
  \[ 1 - \frac{p}{8}\Big(\lambda_{min}(H) + \frac{2}{m}\sum_{\ell=1}^m |\alpha_l|\Big) \,\leq\, \omega^*(H) \,\leq\, 1 - \frac{p}{8}\Big(\lambda_{min}(H) + \frac{2}{m}\sum_{\ell=1}^m |\alpha_l|\Big) + \eta. \]
\end{theorem}

Corollary~\ref{cor:qma} follows from Theorem~\ref{thm:main} by an amplification step described in Section~\ref{sec:amplification}. The proof of the theorem relies on the analysis of the energy test and the energy consistency test, given in Section~\ref{sec:energytest} and Section~\ref{sec:consistencytest} respectively, together with the analysis of the Pauli braiding test given in Section~\ref{sec:twoplayer}. 

\begin{proof}[Proof of Theorem~\ref{thm:main}.]
  First we establish the lower bound. 
  An honest quantum strategy (as described in Definition~\ref{def:honest-strategy} and Definition~\ref{def:honest-h}) acting on an encoded ground state $\ket{\Gamma}$ of $H$, together with an encoding of the additional EPR pairs $\ket{\EPR}^{\otimes (m-1)}_{A'B'}$ required to implement an optimal strategy in the anticommutation game $G_\ac$ (recall we take $G_\ac=\MS$ in this section, so $\wg=1$, $\delta(\eps)=O(\sqrt{\eps})$, and $m=2$) succeeds in the protocol with probability
$ \omega_{\rm honest}(H) = (1-p) + p\, \omega^*_{\energy}(H)$,
where
\begin{align*}
 \omega^*_{\energy}(H) &=\frac{1}{2} + \frac{1}{2}\Big( 1 -
    \frac{1}{4} \lambda_{\min}(H)  - \frac{1}{2m} \sum_{\ell=1}^m |\alpha_{\ell}| \Big)
\end{align*}
denotes the probability of the honest strategy to pass in the energy and consistency tests, each executed with probability $1/2$; the analysis of the energy test is from Lemma~\ref{lem:energy_test}.

Next we establish the upper bound. Suppose a strategy for the players succeeds 
  with overall probability $\omega_{\rm cheat}$, passes the Pauli braiding test with probability $1 -
  \eps$, and passes the energy and consistency tests with probability
  $\omega_{\energy}$; thus $\omega_{\rm cheat} = (1-p)(1 - \eps) + p \,\omega_{\energy}$.
	Applying the combination of Theorem~\ref{thm:gametwoplayer} and Theorem~\ref{thm:isotwoplayer} there exists an $(rn)$-qubit state $\ket{\varphi_1}$ on which the action of the Pauli operators $\sigma_X, \sigma_Z$ is $O(\eps^{1/8})$-consistent with the action of the players' operators $X, Z$ in the cheating strategy. Further, \lemref{energy_consistency} shows that the measurements performed in the energy test are $O(\eps^{d})$-consistent, for some $0<d<1$, with the corresponding product of players' $X$ and $Z$ operators from the Pauli test. Combining these two statements we deduce that an \emph{honest} strategy using the shared state $\ket{\varphi_1}$ will succeed in the Pauli braiding test with probability $1$ (since it is honest and $\wg=1$), and in the energy test with probability at least  $\omega_{\energy} - O(\eps^{d})$.  
	Since this strategy implements valid logical $X$ and $Z$ operators in the energy test, by
  lemma~\ref{lem:energy_test} it passes   the  test
  with probability at most $\omega^*_{\energy}(H)$. Thus
 $   \omega_{\energy} \leq \omega^*_{\energy}(H) + O({\eps^{d}}) $, and
\begin{align*}
    \omega_{\rm cheat} &= (1-p)(1 - \eps) + p\, \omega_{\energy} \\
                   &\leq (1-p)(1- \eps) + p\, \omega_{\energy}^*(H) + O(p\, \eps^{d}) \\
    &\leq \omega_{\rm honest}(H) - (1-p)\eps + O(p\,{\eps^{d}}).
  \end{align*}
  Choosing $p$ to be a sufficiently small constant times $\eta^{1-d}$, 
  for all $0\leq \eps\leq 1$ this  expression is
  less than or equal to $\omega_{\rm honest}(H) + \eta$. 
\end{proof}

\subsection{Analysis of the energy test}
\label{sec:energytest}

The goal of the energy  test is to estimate the energy of a randomly
chosen  term in the Hamiltonian. 

\begin{lemma}
Given a Hamiltonian $H$ as in~\eqref{eq:h-form}, the acceptance probability of the energy test, when the correct Pauli operators are applied by each player on its respective register of the $(rn)$-qubit encoding of an $n$-qubit state $\ket{\psi}$, is
\begin{align*}
\omega^*_{\energy}(H,\ket{\psi}) &= 1-\Big( \frac{1}{2m} \sum_{\ell = 1}^{m} \frac{ |\alpha_{\ell} | +
  \alpha_{\ell} \langle \psi|H_\ell |\psi \rangle}{2} \Big)\\
  &= 1-\Big(\frac{1}{4}\bra{\psi} H\ket{\psi} + \frac{1}{2m }
    \sum_{\ell} |\alpha_{\ell}|\Big),
\end{align*}
where $H_\ell = \sigma_X(a_\ell)\sigma_Z(b_\ell)$ is the $\ell$-th term in the Hamiltonian.
\label{lem:energy_test}
\end{lemma}

\begin{proof}
The proof is a simple calculation in all points similar to that performed in~\cite[Section~4]{ji2015classical}; see in particular the discussion that precedes Theorem~23 in that paper. We omit the details.
\end{proof}

\subsection{Analysis of the consistency test}
\label{sec:consistencytest}

The goal of the energy consistency test is to guarantee that operators used by the special player on $XZ$-type queries are consistent with those used on other types of queries. 

\begin{lemma}
  Suppose the strategy $(N,\ket{\psi})$ for the players succeeds in the energy consistency test and the Pauli braiding test with probability $1 -
  \eps_{}$ each. Then 
  \[\frac{1}{m}\sum_{\ell=1}^m \big
	\| (\sp_\ell - \hat{X}(a)\hat{Z}(b))\ket{\psi}\big\|^2 \,=\,
  O\big(\eps_{}^{1/32}\big), \]
where $a$ and $b$ are strings such that $H_\ell = \sigma_X(a)\sigma_Z(b)$, and $\sp_\ell$ is the observable applied by the special player upon receiving the query $(XZ, a,  b)$ in the energy test. 
 
Moreover, the honest strategy succeeds in the test with probability $1$. 
\label{lem:energy_consistency}
\end{lemma}

\begin{proof}
We show that $XZ$-queries, $X$-queries, and $Z$-queries
  on the special player are all consistent with $\overline{(X, c, c + a)}$ and
  $\overline{(Z, c, c+b)}$  queries to the composite player. The analysis uses similar techniques to the analysis of the linearity test. First, let us
  analyze the case when the verifier chooses $W=X$. Let the POVM applied by the composite player
  be $\{\overline{M}_{c, c+a}^{\alpha \alpha'}\}$ and define marginalized
  operators 
  \[    \overline{M}_{c | c, c+a}^{\alpha} = \sum_{\alpha'} \overline{M}_{c,
    c+a}^{\alpha\alpha'}.\]
Likewise, let the POVM applied by the special player be
$\hat{P}_{\ell}^{\alpha \alpha'}$ and define marginalized operators for the special player: 
\[ \sp^{\alpha}_{a|\ell} = \sum_{\alpha'} \hat{P}^{\alpha \alpha'}_{\ell},
\qquad \sp^{\alpha'}_{b | \ell} = \sum_{\alpha} \hat{P}^{\alpha
  \alpha'}_{\ell}. \]
The observable $\sp_\ell$ corresponding to the product of the special player's
measurement outcomes is defined as
\[ \sp_\ell = \sum_{\alpha \alpha'} (-1)^{\alpha \cdot \alpha'}
\hat{P}^{\alpha \alpha'}_{\ell}. \]
Recall that the Pauli braiding test (\thmref{isotwoplayer})
guarantees the existence of operators $P^A(a,b)$ exactly satisfying
the Pauli relations; let
\[ \Xlin(a) := P^A(s, 0)\qquad\text{and}\qquad \Zlin(b) := P^A(0, b). \]
Item (a) of \thmref{isotwoplayer} guarantees that $\Xlin(a)$ (resp. $\Zlin(b)$) is within $O(\eps_{}^{1/8})$ of 
$\sx(a)$ (resp. $\sz(b)$), in the state-dependent distance $\Drho$. Associated with the observable $\Xlin(a)$ are the projectors
$\Xlin^\alpha(a), \alpha \in \{\pm 1\}$, and likewise $\Zlin^\beta(b)$
for $\Zlin(b)$.

The following relations follow from the
assumption that the players succeed with probability $1-\eps_{}$
in the energy consistency test. We use the notation $\E_{\ell, a \sim
  H_\ell}$ to indicate that the index $\ell$ is chosen uniformly at
random, and then the string $a$ is chosen from the distribution of queries induced
by the Hamiltonian term $H_\ell$; in contrast to $\E_a$ which indicates a uniformly
random string. 
	  \begin{align}
      \E_{\ell, a \sim H_\ell} \E_{c} \CON_\rho(\overline{M}_{c | c, c+a}^\alpha, \sx^\alpha(c)) &= 1 -  O(\eps_{}), \label{eq:energy_cons1}\\
      \E_{\ell, a \sim H_\ell} \E_{c} \CON_\rho(\overline{M}_{c+a | c, c+a}^\alpha, \sx^\alpha(c + a)) &= 1 -   O(\eps_{}), \label{eq:energy_cons2}\\
      \E_{\ell, a \sim H_\ell} \E_c \CON_\rho\Big(\sp_{a | \ell}^{\alpha}, \sum_{\beta \cdot \beta' =
      \alpha}M_{c, c+a}^{\beta \beta'}\Big) &= 1- O(\eps_{}) \label{eq:energy_cons3}.
  \end{align}
  We use these relations to show that the special player's
  marginalized measurement $\sp^{\alpha}_{a|\ell}$ is close to
  $\Xlin^{\alpha}(a)$. We show
  this in two steps. First, we relate the special player's measurement
  $\sp^{\alpha}_{a|\ell}$ to the composite player's measurement:
  \begin{align*}
    \E_{\ell, a \sim H_\ell} \CON_\rho\big(\sp_{a|\ell}^{\alpha}, \Xlin^{\alpha}(a) \big) 
		&\geq     \E_{\ell, a \sim H_\ell}   \E_c \Big[ \CON_\rho\Big(\sum_{\beta \cdot \beta' = \alpha} \overline{M}^{\beta \beta'}_{c, c + \alpha}, \Xlin^\alpha(a) \Big)  - \Drho \Big(\sp^\alpha_{a | \ell}, \sum_{\beta \cdot \beta' = \alpha} \overline{M}^{\beta \beta'}_{c, c+ a}\Big) \Big] \\
    &\geq \E_{\ell, a \sim H_\ell} \E_c   \CON_\rho\Big(\sum_{\beta \cdot \beta' = \alpha} \overline{M}^{\beta \beta'}_{c, c + a}, \Xlin^\alpha(a)\Big) - O(\sqrt{\eps_{}}), 
%    &= \E_{\ell, a\sim H_\ell} \E_c \sum_{\alpha} \sum_{\beta \cdot \beta' = \alpha} \bavg{M^{\beta \beta'}_{c, c+a} \Xlin^\alpha(a)}{\Psi} - O(\sqrt{\eps_{\cons}}).
\end{align*}
where the first inequality follows from Lemma~\ref{lem:approx} and the
second from \myeqref{condist} and~\eqref{eq:energy_cons3}.  Next we relate
$M$ to a product of two measurements $\sx$:
\begin{align*}
    \E_{\ell, a \sim H_\ell} \CON_\rho\big(\sp_{a|\ell}^{\alpha},
  \Xlin^{\alpha}(a) \big)  &\geq \E_{\ell, a\sim H_\ell} \E_c \CON_\rho\Big( \sum_{\beta \cdot \beta' = \alpha} \overline{M}^{\beta}_{c| c, c+a} \overline{M}^{\beta'}_{c+a|c, c+ a}, \Xlin^\alpha(a)\Big)  - O(\sqrt{\eps_{}})\\
  &\geq \E_{\ell, a\sim H_\ell} \E_c \CON_\rho\Big(\sum_{\beta \cdot \beta' = \alpha} \sx^{\beta}(c) \sx^{\beta'}(c+a) \Xlin^\alpha(a)\Big) - O(\sqrt{\eps_{}}),
\end{align*}
as follows from~\eqref{eq:energy_cons2},~\eqref{eq:energy_cons1} and Lemmas~\ref{lem:approx} and~\myeqref{condist}.  Finally, we use the Pauli braiding test to relate
$\sx$ to the exactly linear observable $\Xlin$. Starting from the above and using \lemref{consistency} to switch $\sx(c)$ to $\cx(c)$,
\begin{align*}
  \E_{\ell, a \sim H_\ell} \CON_\rho\Big(\sp_{a|\ell}^{\alpha}, \Xlin^{\alpha}(a) \Big) 
                     \geq \E_{\ell, a\sim H_\ell} \E_c \CON_\rho\Big(
                       \sum_{\beta \cdot \beta' = \alpha}
                       \cx^{\beta}(c) \sx^{\beta'}(c+a),
                       \Xlin^\alpha(a)\Big)
                        - O(\sqrt{\eps_{}}).
												\end{align*}
 Next, we use \thmref{isotwoplayer} and \lemref{approx} to sequentially exchange the remaining $\sx$, then $\overline{X}$, to $\Xlin$, to obtain 
\begin{align}
   \E_{\ell, a \sim H_\ell} \CON_\rho\Big(\sp_{a|\ell}^{\alpha}, \Xlin^{\alpha}(a) \Big) 
	\geq \E_{\ell, a\sim H_\ell} \E_c  \CON_\rho\Big( \sum_{\beta \cdot \beta' = \alpha} \Xlin^{\beta}(c) \Xlin^{\beta'}(c+a) ,\Xlin^\alpha(a)\Big)  - O({\eps}^{1/16}). \label{eq:energytox}
	\end{align}
Finally, the product of the three $\Xlin$ operators can be eliminated using the exact linearity relations. Performing an analogous analysis for the $Z$ operators,
  \begin{equation} \E_{\ell, b \sim H_\ell} \CON_\rho(\sp_{b|\ell}^{\beta}, \Zlin(b)^\beta) \geq 1 - O(\eps_{}^{1/16}). \label{eq:energytoz} \end{equation}
 To put these results together it remains to apply the stabilizer property to these operators. While we cannot do this directly since $a$ and $b$ are not distributed uniformly, we can use the exact linearity to write $\Zlin(b) = \E_{c} \Zlin(b+c) \Zlin(c)$, and apply Lemma~\ref{lem:consistency} to each term in the product:
  \begin{alignat*}{2}
     \sp_\ell \ket{\psi} &=^{\ell} \sp_{a|\ell} \sp_{b|\ell} \ket{\psi} \\
    &\app{\eps^{1/32}}{\ell}  \sp_{a|\ell} \Zlin(b) \ket{\psi} \qquad &\text{by \myeqref{energytoz}, \lemref{approx}, and \myeqref{condist}}\\
    &=^{\ell}\E_c \sp_{a|\ell} \Zlin(b + c) \Zlin(c) \ket{\psi} \qquad &\text{by exact linearity} \\
    &\app{\eps^{1/16}}{\ell} \E_c \cz(c) \cz(b+c) \sp_{a|\ell} \ket{\psi} \qquad &\text{by \thmref{isotwoplayer} and \lemref{consistency}}\\
    &\app{\eps^{1/32}}{\ell} \E_c \cz(c)\cz(b+c) \Xlin(a) \ket{\psi} \qquad &\text{by \myeqref{energytox} and \lemref{approx}}\\
    &\app{\eps^{1/16}}{\ell} \E_c \Xlin(a) \Zlin(b+c)\Zlin(c) \ket{\psi} \qquad &\text{by \thmref{isotwoplayer} and \lemref{consistency}}\\
    &=^{\ell} \Xlin(a) \Zlin(b) \ket{\psi} \qquad &\text{by exact linearity.}
  \end{alignat*}
\end{proof}

\subsection{Amplification}
\label{sec:amplification}

In this section we show how Theorem~\ref{thm:main} can be used to
obtain Corollary~\ref{cor:qma}. The main idea consists in leveraging the fact that our protocol does not require locality of the Hamiltonian to first ``brute-force'' amplify the
gap of the underlying instance of the local Hamiltonian problem to a constant, and then run the protocol on the amplified non-local instance. This is achieved
 by first shifting the Hamiltonian by the appropriate multiple of
identity so that the energy in the yes-instance is less than or equal
to $0$. The gap is amplified by taking sufficiently many tensor product copies of the Hamiltonian, resulting in a nonlocal instance.

\begin{lemma}[Gap amplification]
Let $H$ be an $n$-qubit Hamiltonian with
minimum energy $\lambda_{\min}(H)\geq 0$ and such that $\|H\|\leq 1$. 
  Let $p(n), q(n)$ be polynomials such that $p(n) >
  q(n)$ for all $n$. Let 
	$$H' =  \Id^{\otimes a} - (\Id - ( H - a^{-1}\Id))^{\otimes a},\qquad \text{where}\qquad a = \Big(\frac{1}{q}-\frac{1}{p}\Big)^{-1}.$$
	Then $H'$ is a (non-local) Hamiltonian over $an = O(np(n))$
        qubits with $\|H'\|=O(1)$, such that if $\lambda_{\min}(H) \leq 1/p$, then $\lambda_{min}(H') \leq 1/2$, whereas if $\lambda_{\min}(H) \geq 1/q$, then $\lambda_{min}(H') \geq 1$.
  \label{lem:amplify}
\end{lemma}

\begin{proof}
The proof follows by observing that $\lambda_{\min}(H') = 1 - (1 - (\lambda_{min}(H) - a^{-1}))^a$, and $(1 \pm \delta)^k = 1 \pm k\delta +
  O(\delta^2)$ when $k\delta = O(1)$. 
\end{proof}

\begin{proof}[Proof of Corollary~\ref{cor:qma}]
By applying the result of Theorem~\ref{thm:main} to the Hamiltonian $H'$ obtained from $H$ as in Lemma~\ref{lem:amplify}, we obtain the statement of Corollary~\ref{cor:qma}, except with $p_{\rm c} = p$ and $p_{\rm s} = q $ for some constants $0 < q < p < 1$. To match the constants in the statement of Corollary~\ref{cor:qma}, we make the verifier {automatically} accept with probability $1-p'$, and perform the test with probability $p'$, for some $0\leq p'\leq 1$. Then we get $p_{\rm c} = 1-p' +p'p $ and $p_{\rm s} = 1 - p' + p'q$. If $p'$ is chosen as $p' = \frac{1}{2(1 + p -2q)}$, we get $p_{\rm c} = 1/2 + 2\eta_0$ and $p_{\rm s} = 1/2 + \eta_0$ as desired, with $\eta_0 = \frac{(p-q)}{2(1 + p - 2q)}$.
\end{proof}

\section{Delegated Computation}
\label{sec:delegated}
It was noticed in~\cite{FitzsimonsH15} that an interactive proof system for the local Hamiltonian problem can also be used for delegated quantum computation with so-called \emph{post-hoc} verification. The key idea is to use the Feynman-Kitaev construction to produce a Hamiltonian encoding the desired computation; measuring the ground energy of this Hamiltonian reveals whether the computation accepts or rejects. 
%This connection was used in~\cite{FitzsimonsH15} to give a delegated computation scheme with a limited quantum verifier, with only the capability to send, receive, and measure a constant number of qubits. 
Following the same connection, we are able to give a post-hoc verifiable delegated computation scheme with a purely classical verifier and a constant number of players. The players only need the power of BQP. The scheme has a constant completeness-soundness gap independent of the size of the circuit to be computed, unlike the scheme of~\cite{FitzsimonsH15} and the classical scheme of~\cite{ReichardtUV13nature}, which both have inverse-polynomial gaps. However, unlike the scheme of~\cite{ReichardtUV13nature} (and similarly to the one in~\cite{FitzsimonsH15}), our protocol is not \emph{blind}: the verifier must reveal the entire circuit to be computed to all the players before the verification process starts.

\begin{theorem}
There exists an interactive proof system for BQP with seven quantum entangled players and one classical verifier, with one round of communication, in which the player sends $O(\poly(n))$-bit questions and receives $O(1)$-bit answers. The honest players only need the power of BQP.
\end{theorem}

\begin{proof}[Proof sketch]
  For any poly-size quantum circuit $C$, we construct the history
  Hamiltonian $H_C$ and announce to the seven players. In the honest
  case, the players produce the state
  \[ \ket{\psi}  = \mathrm{ENC}\Big(\frac{1}{\sqrt{T}} \sum_{t = 1}^{T}
  \ket{t}_{\text{clock}} \ot \ket{\psi_t}\Big), \]
  where $\mathrm{ENC}$ is the encoding map of the 7-qubit code, $t$ labels the clock states of the computation from $1$ to
  $T$, and $\ket{\psi_t}$ is the state of the circuit $C$ at step $t$. This
  state can be prepared with a BQP machine. The players are then
  separated; in the honest case, each player receives a share of the
  encoded state $\ket{\psi}$. The verifier plays the game
  of~\thmref{main} with the players and accepts if and only if they succeed.
\end{proof}

\paragraph{Acknowledgments.}

AN was supported by NSF Grant CCF-1629809 and ARO Contract Number W911NF-12-0486. TV was supported by NSF CAREER Grant CCF-1553477 and AFOSR YIP award number FA9550-16-1-0495. 
Parts of this work was completed while the first author was visiting the Institute for Quantum
Information and Matter (IQIM) at the California Institute of Technology, and while both authors were
visiting the Perimeter Institute in Waterloo, ON. Both authors acknowledge funding provided by the IQIM, an NSF Physics Frontiers Center (NFS Grant PHY-1125565) with support of the Gordon and Betty Moore Foundation (GBMF-12500028). 

\bibliography{../quantum_pcp}

\begin{thebibliography}{CHTW04}

\bibitem[AAV13]{AharonovAV13qpcp}
Dorit Aharonov, Itai Arad, and Thomas Vidick.
\newblock The quantum {PCP} conjecture.
\newblock Technical report, arXiv:1309.7495, 2013.
\newblock Appeared as guest column in ACM SIGACT News archive Volume 44 Issue
  2, June 2013, Pages 47--79.

\bibitem[BLR93]{BLR93}
Manuel Blum, Michael Luby, and Ronitt Rubinfeld.
\newblock Self-testing/correcting with applications to numerical problems.
\newblock {\em Journal of Computer and System Sciences}, 47:549--595, 1993.

\bibitem[CHSH69]{CHSH69}
John~F. Clauser, Michael~A. Horne, Abner Shimony, and Richard~A. Holt.
\newblock Proposed experiment to test local hidden-variable theories.
\newblock {\em Phys. Rev. Lett.}, 23:880--884, Oct 1969.

\bibitem[CHTW04]{CHTW04}
Richard Cleve, Peter H{\o}yer, Ben Toner, and John Watrous.
\newblock Consequences and limits of nonlocal strategies, 2004.
\newblock arXiv:quant-ph/0404076.

\bibitem[CM14]{CM13}
Toby Cubitt and Ashley Montanaro.
\newblock Complexity classification of local hamiltonian problems.
\newblock In {\em Foundations of Computer Science (FOCS), 2014 IEEE 55th Annual
  Symposium on}, pages 120--129. IEEE, 2014.

\bibitem[CN16]{CN16}
Matthew Coudron and Anand Natarajan.
\newblock The parallel-repeated {Magic} {Square} game is rigid.
\newblock Technical report, arXiv:1609.06306, 2016.

\bibitem[Col16]{Coladangelo16}
Andrea~W. Coladangelo.
\newblock Parallel self-testing of (tilted) {EPR} pairs via copies of (tilted)
  {CHSH}.
\newblock Technical report, arXiv:1609.03687, 2016.

\bibitem[CS96]{CalderbankShor96}
A.~R. Calderbank and Peter~W. Shor.
\newblock Good quantum error-correcting codes exist.
\newblock {\em Phys. Rev. A}, 54:1098--1105, 1996.
\newblock arXiv:quant-ph/9512032.

\bibitem[Din07]{Dinur07pcp}
Irit Dinur.
\newblock The {PCP} theorem by gap amplification.
\newblock {\em J. ACM}, 54(3), June 2007.

\bibitem[FH15]{FitzsimonsH15}
Joseph Fitzsimons and Michal Hajdu{\v{s}}ek.
\newblock Post hoc verificatin of quantum computing.
\newblock Technical report, arXiv:1512.04375, 2015.

\bibitem[FV15]{FV14}
Joseph Fitzsimons and Thomas Vidick.
\newblock A multiprover interactive proof system for the local {H}amiltonian
  problem.
\newblock In {\em Proceedings of the 2015 Conference on Innovations in
  Theoretical Computer Science}, pages 103--112. ACM, 2015.

\bibitem[Got97]{Gottesman97}
Daniel Gottesman.
\newblock Stabilizer codes and quantum error correction, 1997.
\newblock arXiv:quant-ph/9705052.

\bibitem[IKM09]{IKM09}
Tsuyoshi Ito, Hirotada Kobayashi, and Keiji Matsumoto.
\newblock Oracularization and two-prover one-round interactive proofs against
  nonlocal strategies.
\newblock In {\em Proceedings: Twenty-Fourth Annual IEEE Conference on
  Computational Complexity (CCC 2009)}, pages 217--228, July 2009.

\bibitem[IV12]{IV12}
Tsuyoshi Ito and Thomas Vidick.
\newblock A multi-prover interactive proof for {NEXP} sound against entangled
  provers.
\newblock {\em Proc. 53rd FOCS}, pages 243--252, 2012.

\bibitem[Ji16a]{ji2015classical}
Zhengfeng Ji.
\newblock Classical verification of quantum proofs.
\newblock In {\em Proceedings of the 48th Annual ACM SIGACT Symposium on Theory
  of Computing}, pages 885--898. ACM, 2016.

\bibitem[Ji16b]{ji16nexp}
Zhengfeng Ji.
\newblock Compression of quantum multi-prover interactive proofs.
\newblock Technical report, arXiv:1610.03133, 2016.

\bibitem[KM03]{KM03}
Hirotada Kobayashi and Keiji Matsumoto.
\newblock Quantum multi-prover interactive proof systems with limited prior
  entanglement.
\newblock {\em JCSS}, 66:429--450, 2003.
\newblock arXiv:cs/0102013.

\bibitem[McK14]{mckague2014self}
Matthew McKague.
\newblock Self-testing graph states.
\newblock In {\em Theory of Quantum Computation, Communication, and
  Cryptography}, pages 104--120. Springer, 2014.

\bibitem[McK15]{McKague15}
Matthew McKague.
\newblock Self-testing in parallel, 2015.
\newblock arXiv:1511.04194.

\bibitem[MdW13]{montanaro2013survey}
Ashley Montanaro and Ronald de~Wolf.
\newblock A survey of quantum property testing.
\newblock {\em arXiv preprint arXiv:1310.2035}, 2013.

\bibitem[Mer90]{mermin1990simple}
N~David Mermin.
\newblock Simple unified form for the major no-hidden-variables theorems.
\newblock {\em Physical Review Letters}, 65(27):3373, 1990.

\bibitem[MS14]{MS14}
Carl~A. Miller and Yaoyun Shi.
\newblock Robust protocols for securely expanding randomness and distributing
  keys using untrusted quantum devices.
\newblock In {\em Proceedings of the 46th Annual ACM Symposium on Theory of
  Computing}, STOC '14, pages 417--426, New York, NY, USA, 2014. ACM.

\bibitem[MYS12]{McKagueYS12rigidity}
M~McKague, T~H Yang, and V~Scarani.
\newblock Robust self-testing of the singlet.
\newblock {\em Journal of Physics A: Mathematical and Theoretical},
  45(45):455304, 2012.

\bibitem[NC01]{NieChu01}
Michael~A. Nielsen and Isaac~L. Chuang.
\newblock {\em Quantum Computation and Quantum Information}.
\newblock Cambridge University Press, 2001.

\bibitem[OV16]{OV16}
Dimiter Ostrev and Thomas Vidick.
\newblock Entanglement of approximate quantum strategies in xor games.
\newblock Technical report, arXiv:1609.01652, 2016.

\bibitem[Per90]{peres1990incompatible}
Asher Peres.
\newblock Incompatible results of quantum measurements.
\newblock {\em Physics Letters A}, 151(3-4):107--108, 1990.

\bibitem[RUV13]{ReichardtUV13nature}
Ben Reichardt, Falk Unger, and Umesh Vazirani.
\newblock A classical leash for a quantum system: Command of quantum systems
  via rigidity of {CHSH} games.
\newblock {\em Nature}, 496(7446):456--460, 2013.

\bibitem[Ste96]{Steane96}
Andrew Steane.
\newblock Multiple-particle interference and quantum error correction.
\newblock In {\em Proceedings of the Royal Society of London A: Mathematical,
  Physical and Engineering Sciences}, volume 452, pages 2551--2577. The Royal
  Society, 1996.
\newblock arXiv:quant-ph/9601029.

\bibitem[Vid13]{Vidick13xor}
Thomas Vidick.
\newblock Three-player entangled {XOR} games are {NP}-hard to approximate.
\newblock In {\em Proc. 54th FOCS}, 2013.

\bibitem[VV14]{VV14}
Umesh Vazirani and Thomas Vidick.
\newblock Fully device-independent quantum key distribution.
\newblock {\em Phys. Rev. Lett.}, 113:140501, Sep 2014.

\bibitem[WBMS16]{wu2016device}
Xingyao Wu, Jean-Daniel Bancal, Matthew McKague, and Valerio Scarani.
\newblock Device-independent parallel self-testing of two singlets.
\newblock {\em Physical Review A}, 93(6):062121, 2016.

\end{thebibliography}
\end{document}